\documentclass[11pt,a4paper]{article}

\usepackage[utf8]{inputenc}
\usepackage[T1]{fontenc}
\usepackage[margin=1in]{geometry}
\usepackage{amsmath,amssymb,amsthm,mathtools}
\usepackage{enumitem}
\usepackage{hyperref}
\usepackage{xcolor}
\hypersetup{
  colorlinks=true,
  linkcolor=blue!50!black,
  urlcolor=blue!50!black,
  citecolor=blue!50!black
}

\usepackage{microtype}

\usepackage{tikz}
\usetikzlibrary{positioning, fit, backgrounds}

\theoremstyle{plain}
\newtheorem{theorem}{Theorem}
\newtheorem{proposition}[theorem]{Proposition}
\newtheorem{lemma}[theorem]{Lemma}
\newtheorem{corollary}[theorem]{Corollary}

\theoremstyle{definition}
\newtheorem{definition}[theorem]{Definition}
\newtheorem{example}[theorem]{Example}

\theoremstyle{definition}
\newtheorem{remark}[theorem]{Remark}

\newcommand{\N}{\mathcal{N}}
\newcommand{\I}{\mathcal{I}}
\newcommand{\Lin}{\mathcal{L}}
\newcommand{\Wk}{\mathcal{W}}
\newcommand{\PartN}{\mathrm{Part}(\N)}
\newcommand{\Part}{\mathrm{Part}}
\newcommand{\Sym}{\mathrm{Sym}}
\newcommand{\Stab}{\mathrm{Stab}}

\title{An Axiomatic Theory of Tie-Breaking: Impossibility, Characterization, and Decomposition}
\author{Frank M.\,V.\,Feys}
\date{\today}

\begin{document}
\maketitle

\begin{abstract}
\noindent
We develop an abstract axiomatic theory of tie-breaking. 
A \emph{tie-breaking input} consists of a finite set $\N$ of
players, a weak order on $\N$ representing the standings to be
refined, and an auxiliary information item drawn from a set on which
the symmetric group $\Sym(\N)$ acts. 
Within this minimal framework we prove three theorems. 
First, no tie-breaking rule producing a
strict linear order can be anonymous, provided the input space
contains even one intrinsically symmetric situation, a condition
met in essentially every realistic application. 
Second, when we allow the rule to
output a partition of $\N$ (rather than a strict ranking), there is a
unique rule satisfying two natural axioms: it is the partition of
$\N$ into orbits of the joint stabilizer of the input. 
Third, every reasonable strict tie-breaking rule decomposes uniquely as the
canonical orbit partition followed by an arbitrary completion. 
The decomposition makes precise the informal observation that real
tie-breaking systems are honest until forced to be arbitrary.
The framework is broad enough to capture chess tournament
tie-breakers, sports league regulations, voting tie-breakers,
tie-breaking among symmetric players in cooperative games,
and ranking by network centrality measures, all
within a single uniform formalism.
\end{abstract}

\section{Introduction}
\label{sec:intro}

In June 2010, a tennis match between John Isner and Nicolas Mahut at
Wimbledon reached the score 70--68 in the fifth and final set.
Wimbledon's rules at the time imposed no limit on the score in the
final set, and in the end the match was decided merely by exhaustion. 
After the match, Wimbledon changed its rules to introduce a fixed cap at
12--12 in the final set; once the games reach that score, the
players no longer continue under ordinary scoring but settle the
set with a short playoff, first to seven points. 
The new rule, however, does not really address the underlying difficulty. 
When the score reaches 12--12, the two players are tied: they have won the same
number of points and the same number of games, and the match record
contains no information that singles out one of them as a winner.
The cap stops the match at a fixed score and hands the decision to a playoff, which is a separate contest whose winner is not determined by the preceding match record. This rule change replaces an indefinite continuation with a fixed stopping point, 
but the arbitrariness of selecting a winner from a tied state is simply
relocated, not eliminated.

The Isner--Mahut match is an extreme illustration of a phenomenon that
arises throughout aggregation theory. 
\emph{Strict} tie-breaking, which produces a definitive ranking from data that does not, by itself,
distinguish the players, always requires either residual ties or
arbitrary external data. 
By residual ties we mean ties that the rule
leaves unbroken in its output, accepting that some pairs of players
cannot honestly be separated. 
By arbitrary external data we mean data
not contained in the input itself, such as a coin flip, a
pre-tournament draw, an alphabetical convention on players' names, or
a fixed seeding determined before the contest.
In chess tournaments for example, ties at the top of the
standings are broken by a cascade (Buchholz, Sonneborn--Berger, direct
encounter, etc.), terminated by a coin flip or alphabetical order
when the cascade is exhausted. 
In professional sports leagues, goal difference is followed by goals scored, 
then head-to-head record, then
fair-play points, and eventually a draw. 
In voting, plurality and Borda rules need tie-breakers, and these tie-breakers (lottery, agenda,
fixed voter) all violate fairness in some direction.

A second, more theoretical illustration of the same phenomenon is the Condorcet cycle. 
Indeed, three voters with cyclic preferences given by 
$a \succ b \succ c$, $b \succ c \succ a$, and $c \succ a \succ b$
produce a profile in which the candidates are intrinsically symmetric.
The cyclic permutation $a \mapsto b \mapsto c \mapsto a$ acts on the
candidates and preserves the (anonymous) profile of ballots, since each
ballot type is sent to another ballot type appearing with the same
multiplicity. 
Hence, any anonymous social ranking must place $a$, $b$,
and $c$ as equivalent. 
Yet pairwise majorities are strict and cyclic.
There is no fair way to select a winner.

The Condorcet cycle is the canonical illustration of a celebrated
theorem of Arrow~\cite{arrow1963}; for a comprehensive modern
survey of the field, see the \emph{Handbook of Computational
Social Choice}~\cite{handbook2016}. 
Arrow's impossibility theorem says
that no procedure aggregating individual preferences into a collective
ranking can simultaneously satisfy universal domain, the Pareto
principle, independence of irrelevant alternatives, and
non-dictatorship. 
Roughly stated, there is no good way of aggregating
individual preferences into a collective ranking. 
Arrow's result is the foundational impossibility theorem of modern social choice, and
it inaugurated a tradition in which folk-wisdom observations about
the difficulty of fair aggregation are converted into precise
mathematical statements with structural content. 
The present paper sits in this tradition. 
Our framework is more abstract than Arrow's,
in that it is not specifically about voting; the same theorems apply
uniformly to chess tournaments, sports leagues, voting, cooperative
games, and network centrality. 
Within this framework, the Condorcet cycle is just one example among many.

Our first result, Theorem~A below, is in the same spirit.
Stated at the highest level: \emph{there is no good way to break a tie.}
Just as Arrow's theorem says there is no good way of aggregating
individual preferences into a collective ranking, our impossibility
theorem says there is no good way of refining a partial ranking into a strict one.
Both results have the same structural shape, namely a
small set of natural desiderata (anonymity, refinement, and
strictness in our case; universal domain, Pareto, independence of
irrelevant alternatives, and non-dictatorship in Arrow's) that turn
out to be jointly inconsistent on any reasonable input.
As with Arrow's theorem, ``no good way'' should be read with care.
It means ``no rule satisfying these particular axioms,'' not ``no rule whatsoever.''
Practitioners use tie-breaking rules every day, just as
practitioners run elections every day.
What both theorems pin down is exactly which desideratum must be given up, 
and the structure of what remains.
Theorems~B and~C below make this structural content precise.
The orbit partition $\Omega$ is the canonical, forced part of any
honest tie-breaking, and any strict tie-breaking rule must supply, in
addition, an irreducibly arbitrary completion.

The examples mentioned above are diverse in surface but share a single underlying structure.
In each, an aggregation procedure produces a coarse weak order on a set of players or alternatives. 
Some pairs of players are tied in this coarse order, and the practical task is to refine the order further. 
 The auxiliary data available to refine it (game records, voter ballots, 
characteristic-function values) carries an  intrinsic symmetry group. 
When two tied players are exchanged by a symmetry of this auxiliary data, 
no honest refinement can distinguish them, 
and the practitioner is therefore forced either to leave the tie unbroken
or to introduce some form of arbitrary external data.

The aim of this paper is to give a clean, abstract, axiomatic account of this phenomenon. 
We work in a minimal framework: a finite player
set $\N$, a weak order $\succeq$ on $\N$ representing the standings to
be refined, and an auxiliary information item $h$ drawn from a set
$\I$ on which the symmetric group $\Sym(\N)$ acts. 
A tie-breaking rule is a function $T=T(h, \succeq)$ producing some refinement of $\succeq$. 
The framework is deliberately structure-free. 
In fact, $\I$ has no internal structure beyond the $\Sym(\N)$-action. 
As we shall show, this minimal setting is rich enough to capture every concrete tie-breaking
practice we are aware of, and yet abstract enough to admit clean
theorems about all of them at once.

A reader may reasonably ask why we work at this level of abstraction.
The answer, in fact, is empirical. 
The same obstruction to strict tie-breaking appears in chess Swiss-system tournaments, in round-robin sports
leagues, in voting on a cycle of candidates, in cooperative games
with symmetric players, and in network centrality on graphs with
non-trivial automorphism groups, among others. 
Without a common framework, these are five separate observations, 
each established by methods specific to its setting. 
With a common framework, however, they are five instances of just one theorem. 
The abstraction is what allows the impossibility, the
characterization, and the decomposition to be stated and proved once, then applied uniformly. 
Section~\ref{subsec:examples} works through
all five settings as instances of the framework, and we view this
collection of examples as the empirical case for the abstraction.

We prove three theorems.

 \medskip

\noindent\textbf{Theorem A (Impossibility; Theorem~\ref{thm:impossibility} below).} \textit{Strict tie-breaking is impossible.} 
No rule $T \colon \I \times \Wk(\N) \to \Lin(\N)$ producing
linear orders can be anonymous, provided the input space contains
even a single \emph{symmetric input}: 
an input $(h, \succeq)$ admitting two distinct players that some intrinsic symmetry of $(h, \succeq)$ exchanges. 
Symmetric inputs exist in essentially every realistic instance of the framework,
so the impossibility applies essentially everywhere.

 \medskip

\noindent\textbf{Theorem B (Characterization; Theorem~\ref{thm:characterization} below).} \textit{Partition-valued tie-breaking has a unique honest answer.} 
If we allow the rule to output a partition of $\N$ (a record of who the rule declares
``still tied''), the unique rule satisfying two natural axioms,
\emph{symmetry-saturation} and \emph{maximal-fineness}, is
$T(h, \succeq) = \Omega(h, \succeq)$, where $\Omega$ is the partition
into orbits of the joint stabilizer of $(h, \succeq)$.

 \medskip

\noindent\textbf{Theorem C (Decomposition; Theorem~\ref{thm:decomposition} below).} \textit{Every reasonable strict tie-breaking rule 
decomposes uniquely as canonical core plus arbitrary completion.} 
Given any strict rule $T$ that respects the
canonical partition (in a sense made precise later), $T$ corresponds
bijectively to a choice of completion: a linear order within each
block of $\Omega$, together with a linear order on the blocks of
$\Omega$ within each indifference class of $\succeq$. 
The canonical part is forced; the completion is the arbitrary part.

\medskip

The three theorems together capture the phenomenon. 
Strict tie-breaking is impossible (Theorem A); 
the partition-valued analogue exists uniquely (Theorem B); 
and every reasonable strict tie-breaker
is the canonical partition plus arbitrary linearizing data (Theorem C). 
The decomposition theorem, in particular, makes precise
the informal observation that \emph{tie-breaking systems used in practice are honest
until forced to be arbitrary}. 
The canonical part records exactly the
structure the input determines; the completion records exactly the
data that must be supplied externally.

\subsection*{Outline}

The paper is organized as follows. 
Section~\ref{sec:framework} sets
up the abstract framework, with a substantial subsection of worked
examples illustrating the breadth of the formalism.
Section~\ref{sec:imp} proves the impossibility theorem.
Section~\ref{sec:partition} introduces partition-valued rules and
proves the characterization theorem. 
Section~\ref{sec:decomp} proves the decomposition theorem and revisits two of the worked examples.
Section~\ref{sec:weak} gives the weak-order perspective, for
readers who prefer that language. 
Section~\ref{sec:related} discusses related literature. 
Section~\ref{sec:conclusion} concludes with directions for further work.

\section{Framework}
\label{sec:framework}

We fix throughout a finite set $\N$ with $|\N| \geq 2$. 
Its elements are called \emph{players}. 
The use of ``player'' is not meant to restrict the framework to game-theoretic settings; 
it is a neutral term standing in for ``alternatives'' in voting, ``teams''
in sports, ``contestants'' in tournaments, 
or any other interpretation
the application in question calls for.

\subsection{Weak Orders, Partitions, Refinement}

We begin by fixing notation for the order-theoretic and combinatorial
objects that will appear repeatedly. 
The two basic notions are weak orders, which model standings that may have ties, and partitions,
which record the equivalence structure of these ties. 
The relation between them is captured by the indifference partition map, which
sends a weak order to its partition into ``ranked-equally'' classes.
The notion of refinement gives a partial order on partitions, with
finer partitions sitting below coarser ones, and the indifference
partition map is order-preserving in a natural sense.

A \emph{partition} of $\N$ is a set $\mathcal{P}$ of pairwise disjoint
nonempty subsets covering $\N$. 
The members of $\mathcal{P}$ are called \emph{blocks}. 
We write $\PartN$ for the set of partitions of
$\N$. For $\mathcal{P}, \mathcal{Q} \in \PartN$, we say that
$\mathcal{P}$ \emph{refines} $\mathcal{Q}$, written
$\mathcal{P} \trianglelefteq \mathcal{Q}$, if every block of
$\mathcal{P}$ is contained in some block of $\mathcal{Q}$. 
Refinement is a partial order on $\PartN$, with top $\{\N\}$ (the trivial
partition, one block) and bottom $\{\{x\} \mid x \in \N\}$ (the discrete
partition, all singletons).

A \emph{weak order} on $\N$ is a reflexive, transitive, total binary relation. 
We write $\Wk(\N)$ for the set of weak orders on $\N$, and
$\Lin(\N)$ for the subset of \emph{linear orders} (antisymmetric weak
orders).

For ${\succeq} \in \Wk(\N)$, the relation
$x \sim y \iff (x \succeq y$ and $y \succeq x)$ is an equivalence
relation; its equivalence classes are the \emph{indifference classes} of $\succeq$. 
We denote by $\Part(\succeq)$ the partition of $\N$
into indifference classes of $\succeq$.

\subsection{The Canonical $\Sym(\N)$-Actions}
\label{subsec:Sym-actions}

The framework rests on the following idea: the players in $\N$ are
interchangeable from the standpoint of the rule we eventually construct. 
Concretely, we should be able to relabel them, and any
honest rule must produce the relabeled output when given relabeled input. 
To formalize this, we need a careful account of what it means
to relabel each kind of object that appears in the framework: 
weak orders, partitions, and (later) auxiliary information. 
Each of these has a natural action of the symmetric group $\Sym(\N)$, and these
actions are compatible with one another in a precise sense. 
We make this explicit here.

Note that the symmetric group $\Sym(\N)$ acts naturally on weak orders and on partitions of $\N$. 
For $\sigma \in \Sym(\N)$ and
${\succeq} \in \Wk(\N)$, define $\sigma \cdot {\succeq}$ by
\[
   x \,(\sigma \cdot {\succeq})\, y
   \quad \iff \quad
   \sigma^{-1}(x) \succeq \sigma^{-1}(y).
\]
For $\sigma \in \Sym(\N)$ and $\mathcal{P} \in \PartN$, define
\[
   \sigma \cdot \mathcal{P}
   :=
   \{\sigma(B) \mid B \in \mathcal{P}\}.
\]
Both are group actions in the standard sense: 
the identity acts as the identity, and
$   (\sigma\tau) \cdot x = \sigma \cdot (\tau \cdot x)$
for all $\sigma, \tau \in \Sym(\N)$ and all $x$ in the action set.
Furthermore, the indifference partition map $\Part \colon \Wk(\N)
\to \PartN$ is equivariant:
\[
   \Part(\sigma \cdot {\succeq})
   =
   \sigma \cdot \Part({\succeq}).
\]
This identity follows from the observation that
$x \sim_{\sigma \cdot {\succeq}} y$ 
if and only if
$\sigma^{-1}(x) \sim_{\succeq} \sigma^{-1}(y)$, so the indifference
classes of $\sigma \cdot \succeq$ are exactly the $\sigma$-images of
the indifference classes of $\succeq$.

\subsection{Information Space and Inputs}

We come now to the  central modeling decision of the paper. 
A tie-breaking rule must work with two ingredients. 
First, the standings to be refined, given by a weak order 
$\succeq$ on $\N$, with possibly some players tied. 
Second, the auxiliary data the rule is allowed to
consult, such as the tournament's game record, the voters' ballots, the
characteristic function of the cooperative game, or the graph encoding the network. 
This auxiliary data takes wildly different mathematical forms across applications. 
A chess tournament's  pairing-and-result
record is structurally nothing like a voting profile or a
characteristic function or a graph. 
The framework must accommodate all of them at once.

We adopt the most economical solution possible. 
The auxiliary data is modeled as an abstract set $\I$, equipped with one and only one
piece of structure: 
an action of  the symmetric group $\Sym(\N)$.
The action records what it means to relabel players, which is the
single ingredient our results require. 
No further structure (algebraic operations, distances, topology) is assumed, 
and as we shall see, none is needed. 
This minimalism is what makes the framework apply uniformly to such diverse settings.

\begin{definition}[Information Space]
\label{def:info-space}
An \emph{information space} on $\N$  is a $\Sym(\N)$-set, i.e., 
 a set $\I$
together with an action
\[
   \Sym(\N) \times \I \to \I,
   \qquad
   (\sigma, h) \mapsto \sigma \cdot h.
\]
\end{definition}

The information space is structure-free in the sense that no further
algebraic, topological, or order-theoretic structure is assumed
beyond the group action. 
The single requirement is that whatever
auxiliary data the practitioner has on $\N$, this data must be
relabelable when players are relabeled, and the relabeling must
respect group composition. 
This is a remarkably mild requirement,
and we shall see in Section~\ref{subsec:examples} that it is satisfied by
every concrete tie-breaking input we are aware of.

\begin{definition}[Input]
\label{def:input}
An \emph{input} is an element of $\I \times \Wk(\N)$. 
The diagonal $\Sym(\N)$-action on inputs is
\[
   \sigma \cdot (h, \succeq)
   :=
   (\sigma \cdot h, \sigma \cdot \succeq).
\]
\end{definition}

Intuitively, an input is the full data the tie-breaking rule has to work with: 
an information item $h$ (the auxiliary data, that is, the
tournament's games, the voters' ballots, the characteristic
function, etc.) and a coarse weak order $\succeq$ (the standings to be refined). 
The diagonal action says that when we relabel the players,
both the auxiliary data and the standings get relabeled together.

We fix an information space $\I$ throughout the paper.

\subsection{Joint Stabilizer}

A relabeling $\sigma \in \Sym(\N)$ may or may not change a given input. 
If $\sigma$ leaves both the auxiliary data $h$ and the
standings $\succeq$ unchanged, then from the input's standpoint, the
permutation $\sigma$ is invisible. 
The input ``looks the same'' under relabeling by $\sigma$. 
Such permutations capture the input's
intrinsic symmetries; in fact, they form the central object of the framework.

\begin{definition}[Joint Stabilizer]
\label{def:joint-stab}
For an input $(h, \succeq)$, the \emph{joint stabilizer} is
\[
   \Stab(h, \succeq)
   :=
   \{\sigma \in \Sym(\N)
   \mid
   \sigma \cdot h = h
   \text{ and }
   \sigma \cdot \succeq \,=\, \succeq   \}.
\]
\end{definition}

It is straightforward to verify that $\Stab(h, \succeq)$ is a
subgroup of the symmetric group $\Sym(\N)$. 
Intuitively,
$\sigma \in \Stab(h, \succeq)$ is a relabeling of the players under
which the input is indistinguishable from itself: both $h$ and
$\succeq$ look exactly the same after applying $\sigma$.

The intuition deserves a further word. 
Indeed, if a non-trivial $\sigma \in \Stab(h, \succeq)$ exchanges two players $x$ and $y$,
then the input contains no asymmetric information distinguishing $x$ from $y$. 
Whatever can be said of $x$ in the input's vocabulary can
equally be said of $y$, and vice versa. 
In other words,  $x$ and $y$ are indistinguishable from the input's standpoint. 
As we shall see, this is exactly the kind of indistinguishability that obstructs
strict tie-breaking and that the canonical refinement we construct
will respect.

The following structural lemma is the technical workhorse of the paper. 
It establishes that any symmetry of $\succeq$ preserves each
indifference class as a set, which in particular implies that any
joint symmetry of $(h, \succeq)$ does as well.

\begin{lemma}[Stabilizers Preserve Indifference Classes]
\label{lem:stab-preserves-classes}
Let ${\succeq} \in \Wk(\N)$ and $\sigma \in \Sym(\N)$. 
If $\sigma \cdot \succeq \,=\, \succeq$, then $\sigma(E) = E$ for every
indifference class $E$ of $\succeq$.
\end{lemma}

\begin{proof}
Let $E_1, \ldots, E_k$ be the indifference classes of $\succeq$,
ordered by the strict-class relation $\succeq$ induces on classes,
with $E_i \succ E_j$ for $i < j$ (this strict order is indeed a linear order on
classes since $\succeq$ is total).
Since $\Part$ is equivariant and
$\sigma \cdot \succeq \,=\, \succeq$, we have
$\Part(\succeq) = \Part(\sigma \cdot \succeq)
= \sigma \cdot \Part(\succeq)$,
i.e., $\sigma$ permutes the classes $E_1, \ldots, E_k$ among themselves. 
Let $\pi$ be the induced permutation of indices, so
$\sigma(E_i) = E_{\pi(i)}$.

We now show $\pi$ is the identity. 
Pick $x \in E_i$ and $y \in E_j$
with $i < j$, so $x \succ y$. 
Apply $\sigma$ to both sides. 
By the definition of the action,
$\sigma(x) \,(\sigma \cdot \succeq)\, \sigma(y)$. 
Since
$\sigma \cdot \succeq \,=\, \succeq$, we obtain
$\sigma(x) \succ \sigma(y)$. 
But $\sigma(x) \in E_{\pi(i)}$ and
$\sigma(y) \in E_{\pi(j)}$, so $\pi(i) < \pi(j)$. 
Thus, $\pi$ is an
order-preserving permutation of $\{1, \ldots, k\}$, hence the identity. 
Therefore, $\sigma(E_i) = E_i$ for all $i$.
\end{proof}

\begin{corollary}
\label{cor:orbits-in-classes}
For every input $(h, \succeq)$ and every $x \in \N$, the
$\Stab(h, \succeq)$-orbit of $x$ is contained in the indifference
class of $x$.
\end{corollary}

\begin{proof}
Every $\sigma \in \Stab(h, \succeq)$ satisfies
$\sigma \cdot \succeq \,=\, \succeq$, and hence, by
Lemma~\ref{lem:stab-preserves-classes}, it preserves each indifference
class as a set. 
So if $x$ lies in indifference class $E$, then
$\sigma(x) \in E$ for all $\sigma \in \Stab(h, \succeq)$, and 
thus the orbit of $x$ is contained in $E$.
\end{proof}

The corollary shows that the $\Stab(h, \succeq)$-orbits on $\N$
refine the indifference classes of $\succeq$, and so they form a
partition of $\N$ that records exactly which players are
interchangeable under some intrinsic symmetry of the input. 
We shall denote this partition by $\Omega(h, \succeq)$ and call it the
\emph{orbit partition} of the input. 
Two players $x$ and $y$ lie in
the same block of $\Omega(h, \succeq)$ if and only if some
$\sigma \in \Stab(h, \succeq)$ satisfies $\sigma(x) = y$, that is, 
some intrinsic symmetry of $(h, \succeq)$ exchanges them.
The orbit partition is the central object of the partition framework
developed in Section~\ref{sec:partition}, where the formal definition
is given and where the characterization theorem
(Theorem~\ref{thm:characterization}) singles it out as the unique
honest partition-valued tie-breaking rule; 
we already introduce it informally here so that 
the worked examples below can meaningfully refer to it.

\subsection{Examples}
\label{subsec:examples}

We now illustrate the framework with five worked 
examples drawn from diverse application areas. 
The aim is to demonstrate that the
abstract framework genuinely captures concrete tie-breaking
situations across very different mathematical settings.

\subsubsection*{Example 1: Chess Swiss-System Tournament}

Let $\N$ be the set of players in a chess Swiss-system tournament with $r$ rounds. 
After all rounds, each player has a primary score
(games won, with draws counting as half). 
Group players by primary score; 
players within a group are tied at the level of $\succeq$.
The standings $\succeq$ is the weak order in which players in
higher-scoring groups are strictly above players in lower-scoring
groups, and players within a group are tied.

The auxiliary information $h$ records the pairing-and-result data:
which player played which player in which round, with what color
and what outcome. 
Formally, we may take $h$ to be a function
\[
   h \colon \N \times \N \times \{1, \ldots, r\}
       \to
       \{\mathrm{n}, \mathrm{w}, \mathrm{d}, \mathrm{l}\}
       \times \{\mathrm{white}, \mathrm{black}, -\},
\]
where the first coordinate of the output records the result for the
first player ($\mathrm{n}$ for ``no game played,'' $\mathrm{w}$ for
``win,'' $\mathrm{d}$ for ``draw,'' $\mathrm{l}$ for ``loss''), and
the second coordinate records the color assignment (white, black,
or $-$ to indicate ``no game played'').
The space $\I$ consists of
those functions $h$ that are \emph{consistent} in the following
sense: if $h(x, y, t) = (\mathrm{w}, \mathrm{white})$, then
$h(y, x, t) = (\mathrm{l}, \mathrm{black})$, and similarly for the
other paired cases (loss/win with reversed colors, draw/draw with
reversed colors, no-game/no-game). 
Concretely, the value of $h(x, y, t)$ records the outcome and color for player $x$ in the
game played between $x$ and $y$ in round $t$. 
If players $x$ and $y$ were not paired in round $t$, then $h(x, y, t) = (\mathrm{n}, -)$.
The space $\I$ is a $\Sym(\N)$-set under the natural relabeling action. 
For $\sigma \in \Sym(\N)$, we have 
$(\sigma \cdot h)(x, y, t) := h(\sigma^{-1}(x), \sigma^{-1}(y), t)$.
The consistency condition is preserved by the action.

A symmetry of $(h, \succeq)$ is a relabeling $\sigma \in \Sym(\N)$
that preserves both the score-group structure and the entire game record. 
Two tied players $x, y$ are exchanged by such a
$\sigma$ (i.e., $\sigma(x) = y$ and $\sigma(y) = x$) when there
exists a permutation of the entire player set, restricting to a
swap on $\{x, y\}$, under which the whole game record is invariant.
Such a permutation may also act non-trivially on other players. 
The orbit partition $\Omega(h, \succeq)$ records exactly this
indistinguishability structure. 
Two tied players land in the same
block of $\Omega$ if and only if some such $\sigma \in
\Stab(h, \succeq)$ exchanges them.

This example is a prototype for the entire family of Swiss-system
tie-breakers (Buchholz, Sonneborn--Berger, Median, etc.). 
Each such tie-breaker is a function of $h$ alone, designed to extract some
distinguishing feature when one exists. 
The framework's perspective
is that no such function can distinguish two players whose game
histories are symmetric to each other, and the orbit partition is
the canonical record of which distinctions are honestly available.

\subsubsection*{Example 2: Round-Robin Tournament}

Let $\N$ be the set of teams in a single round-robin tournament,
in which each pair of teams plays exactly once. 
The standings $\succeq$ is given by points (typical scoring: 3 for a win, 1 for a
draw, 0 for a loss). 
Two teams with the same points are tied in
$\succeq$.

The auxiliary information $h$ records the full match-result matrix,
including goal differences. 
We take $h$ to be a function
\[
   h \colon \N \times \N
       \to
       \mathbb{Z} \cup \{\bot\},
\]
where, for distinct $x, y \in \N$, the value $h(x, y) \in
\mathbb{Z}$ is the goal difference for player $x$ in the match
between $x$ and $y$ (goals scored by $x$ minus goals conceded), and
the diagonal value $h(x, x) := \bot$ is a designated symbol
indicating ``not applicable'' (a player does not play against itself). 
The order of arguments $(x, y)$ is a bookkeeping
convention indicating from whose perspective the goal difference is
measured; it does not encode home or away. 
The same match between
distinct teams $x$ and $y$ is recorded twice, once as $h(x, y)$ and
once as $h(y, x)$, with the consistency condition
$ h(x, y) = -h(y, x) $ for all distinct $x, y \in \N.$
The space $\I$ consists of all such consistent functions, and is a
$\Sym(\N)$-set under simultaneous row-and-column relabeling.

Let us illustrate with the classical ``rock-paper-scissors''
configuration, which exhibits a three-way intrinsic symmetry that
no goal-difference-style tie-breaker can resolve. 
Suppose $\N$ contains at least three teams $a, b, c$, and consider an input in
which:
\begin{itemize}[leftmargin=2em]
\item the standings $\succeq$ place $a, b, c$ in a common top
indifference class, with all other teams strictly below;
\item among $a, b, c$, the head-to-head results form a cycle:
$a$ beats $b$ by some goal-difference margin $d > 0$, $b$
beats $c$ by $d$, and $c$ beats $a$ by $d$, so
$h(a, b) = h(b, c) = h(c, a) = d$;
\item against every other team $z \in \N \setminus \{a, b, c\}$,
the three top teams have identical results:
$h(a, z) = h(b, z) = h(c, z)$.
\end{itemize}
Let $\rho \in \Sym(\N)$ be the cyclic permutation
$a \mapsto b \mapsto c \mapsto a$, fixing every player outside
$\{a, b, c\}$. We verify that $\rho \in \Stab(h, \succeq)$. 
The standings $\succeq$ are fixed by $\rho$ because $\rho$ permutes
$\{a, b, c\}$ within their common indifference class and fixes
every other player. 
The matrix $h$ is fixed by $\rho$ because each
match result is mapped to another match result with the same goal
difference: among $\{a, b, c\}$, the entries $h(a, b)$, $h(b, c)$,
$h(c, a)$ are permuted cyclically and all equal $d$, and the
entries $h(b, a)$, $h(c, b)$, $h(a, c)$ are permuted cyclically and
all equal $-d$ by consistency; the entries involving an outside
team $z$ are permuted among the three indistinguishable values
$h(a, z) = h(b, z) = h(c, z)$ and their negatives.

Since $\rho$ acts cyclically on $\{a, b, c\}$ with no fixed point,
the orbit partition $\Omega(h, \succeq)$ has $\{a, b, c\}$ as a single block.
 No tie-breaker computed from $h$ alone can
distinguish the three top teams, since any such tie-breaker is
invariant under $\rho$ and therefore assigns $a, b, c$ the same value. 
Real-world football regulations, when they encounter such
configurations, fall back on coin flips, lots, or a pre-tournament
draw position, which is to say on data outside the match record
itself.

\subsubsection*{Example 3: Voting and the Condorcet Cycle}

Let $\N$ be the set of alternatives in a voting setting (so $\N$ is the set of candidates). 
The standings $\succeq$ is the output of
some primary aggregation procedure (Borda, Copeland, etc.) applied
to the voters' ballots. 
When the primary procedure produces ties,
the practical task is to refine $\succeq$.

The auxiliary information $h$ is the voter profile (anonymous
version): a multiset of weak orders on $\N$, one for each voter.
Formally, $\I$ is the set of functions
$\Wk(\N) \to \mathbb{Z}_{\geq 0}$ with finite support, the
``histogram'' interpretation of profiles. The $\Sym(\N)$-action
relabels candidates (not voters, who are anonymous): for
$\sigma \in \Sym(\N)$,
$(\sigma \cdot h)({\succeq'}) := h(\sigma^{-1} \cdot {\succeq'})$
for each weak order ${\succeq'}$.

The Condorcet cycle is in fact the canonical symmetric input. 
Take $\N = \{a, b, c\}$ and three voters, each contributing one ballot:
\begin{itemize}[leftmargin=2em]
\item one voter votes $a \succ b \succ c$,
\item one voter votes $b \succ c \succ a$,
\item one voter votes $c \succ a \succ b$.
\end{itemize}
The cyclic permutation $\rho = (a\,b\,c)$ maps the above profile to
itself, since each ballot type is permuted to another ballot type also
present in the histogram, with multiplicities preserved. 
So $\rho \in \Stab(h, \succeq_\top)$, where $\succeq_\top$ is any
candidate-tied weak order, in particular the all-tied one. 
Since $\rho$ has no fixed point in $\{a, b, c\}$, the orbit partition is
$\Omega(h, \succeq_\top) = \{\{a, b, c\}\}$, so all three candidates
collapse into a single block. 
There is no honest way to break the
tie within $\{a, b, c\}$ from the profile data alone.

\subsubsection*{Example 4: Cooperative Game with Symmetric Players}

This example shows that the framework's notion of intrinsic
symmetry recovers the classical symmetry axiom of cooperative game theory~\cite{shapley1953}. 
Let $\N$ be the set of players in a cooperative game with
characteristic function $v \colon 2^\N \to \mathbb{R}$. 
The standings $\succeq$ might be the order on Shapley values, or on
some other power index. 
When two players have equal Shapley values,
they are tied in $\succeq$.

The auxiliary information is $h := v$ itself, viewed as an element
of $\I = \mathbb{R}^{2^\N}$. The $\Sym(\N)$-action is
$(\sigma \cdot v)(S) := v(\sigma^{-1}(S))$ for each coalition $S$.

Two players $i, j$ are called \emph{symmetric in $v$} if
$v(S \cup \{i\}) = v(S \cup \{j\})$ for every coalition $S$ not
containing either. 
This standard cooperative-game-theoretic notion
turns out to be exactly the condition that the transposition $(i\,j)$ fixes $v$. 
Hence symmetric players in $v$ are members of the same $\Stab(v)$-orbit. 
When $\succeq$ is the Shapley-value order, the
classical ``symmetry axiom'' (symmetric players have equal Shapley
values) ensures that symmetric players are also tied in $\succeq$,
and the orbit partition $\Omega(v, \succeq)$ groups symmetric
players into the same block. 
No honest function of $v$ alone can distinguish them.
The orbit partition thus extends the classical
pairwise notion of symmetric players to a canonical
partition-valued tie-breaking rule, applicable across all
cooperative games at once.

\subsubsection*{Example 5: Network Centrality on a Symmetric Graph}

Suppose we want to rank the vertices of a network by their structural importance. 
This task arises in many real settings:
ranking influence in social networks, identifying critical nodes
for the security or robustness of an infrastructure system,
ordering scientific papers by citation-based importance, or
producing an authority ranking for web pages. 
The standard approach
is to compute a \emph{centrality measure} on the graph, by which we
mean a function that assigns each vertex a real number expressing
how central it is. 
Common centrality measures include degree (the
number of neighbors), closeness (the reciprocal of the average
distance to other vertices), betweenness (the frequency with which
a vertex lies on shortest paths between other pairs of vertices),
and eigenvector centrality (a vertex is important if it is
connected to important vertices, the principle behind PageRank).
Each centrality measure induces a weak order on the vertices, and
tied vertices need a tie-breaker.

Take an undirected simple graph $G$ on the vertex set $\N$. 
The standings $\succeq$ is the weak order induced by some
centrality measure on $G$. 
The auxiliary information is the graph
itself, $h := G \in \I$, where $\I$ is the set of undirected simple
graphs on $\N$. 
The $\Sym(\N)$-action is the natural relabeling
action; that is, $\sigma \cdot G$ is the graph obtained from $G$ by
relabeling each vertex $x$ as $\sigma(x)$. 
The stabilizer $\Stab(G)$ of $G$ alone is the automorphism group $\mathrm{Aut}(G)$
in the classical graph-theoretic sense.

A centrality measure is \emph{permutation-invariant} if the
centrality value of a vertex depends only on the structural role of
that vertex in the graph and not on its label. 
Equivalently, for every $\sigma \in \Sym(\N)$, the centrality of $\sigma(x)$ in
$\sigma \cdot G$ equals the centrality of $x$ in $G$. 
All standard centrality measures (degree, closeness, betweenness, eigenvector
centrality, and so on) are permutation-invariant. 
For any such measure, every graph automorphism preserves the centrality-induced
ranking, since automorphisms map each vertex to a vertex of equal centrality. 
Hence, $\Stab(G, \succeq) = \mathrm{Aut}(G)$, and the
orbit partition $\Omega(G, \succeq)$ is the decomposition of $\N$
into automorphism orbits of $G$. 
Two vertices in the same orbit
are indistinguishable by any permutation-invariant centrality
measure.

The Petersen graph~\cite{petersen1898} is a canonical example. 
It is a $3$-regular graph on ten vertices whose
 automorphism group acts transitively on the vertex set. 
Hence, $\Omega(G, \succeq)$ has a single block
containing all ten vertices, regardless of which
permutation-invariant centrality measure is used to define $\succeq$. 
There is no honest way to break ties among the ten vertices.

This example shows that the classical graph-theoretic notion of
automorphism group is recovered as a special case of the joint
stabilizer $\Stab$, and that $\Omega$ extends the classical orbit
decomposition of vertices to a partition-valued tie-breaking rule
applicable across all permutation-invariant centrality measures
at once.

\subsubsection*{Common Structure}

These five examples come from very different mathematical settings:
discrete game records, integer-valued result matrices, weak-order
histograms, real-valued set functions, and graphs. 
Yet each is naturally a $\Sym(\N)$-set, each yields a meaningful joint
stabilizer, and in each the orbit partition $\Omega$ corresponds to
a concrete and intuitive notion of ``intrinsically indistinguishable.'' 
The framework's strength is precisely that
this common structure is captured by a single formalism, and that
the three theorems we prove below apply uniformly to all of them.

\section{The Impossibility of Strict Anonymous Tie-Breaking}
\label{sec:imp}

What does the impossibility theorem that we are about to prove actually say, in plain terms? 
It says that any tie-breaking system aiming for
a strict, definitive ranking must, on at least some inputs, reach
outside the data it has been given (such as coin flips, lotteries,
alphabetical conventions, fixed seedings).
 In any honest accounting, a strict tie-breaker eventually appeals to something of this kind.
There is no way to avoid the appeal by being clever about how the
rule is constructed.

Folk wisdom about tie-breaking has long held that any cascade of
tie-breaking criteria must, in the end, bottom out in something arbitrary. 
Theorem~A converts this folk wisdom into a rigorous,
unconditional theorem. 
Indeed, the conclusion is not asymptotic, not
worst-case, and not parametric.
 A single symmetric input anywhere in the
input space is enough to force the impossibility, and no scheme,
however ingenious, evades it on that input.

We now turn to the proof. 
The argument is short and turns on the
inability of a non-trivial permutation to fix a linear order.

\subsection{Strict Tie-Breaking Rules and Anonymity}

The most natural goal for a tie-breaking rule is, at the same time, the most ambitious. 
We aim to produce a definitive, complete ranking with no remaining
ties, and we call such a rule \emph{strict}. 
The strictness requirement is the source of the impossibility, since, as we shall see,
demanding a strict output forces the rule to make distinctions that the input
itself cannot justify.

A strict rule must respect two natural compatibility conditions.
The first is \emph{refinement}, meaning the rule cannot reverse strict
comparisons that
$\succeq$ already settles, i.e., 
if $\succeq$ ranks $x$ above $y$, so must the rule too.
The second is \emph{anonymity}, meaning the rule treats
players uniformly, so that relabeling the
 input simply relabels the output. 
We define both formally and then turn to the 
central technical notion of a
symmetric input.

\begin{definition}[Strict Tie-Breaking Rule]
\label{def:strict-rule}
A \emph{strict tie-breaking rule} is a function
\[
   T \colon \I \times \Wk(\N) \to \Lin(\N)
\]
such that, for every input $(h, \succeq)$, the linear order
$T(h, \succeq)$ refines $\succeq$ as a weak order, namely, for every
$x, y \in \N$ with $x \succ y$ strictly in $\succeq$ it holds that 
$x >_{T(h, \succeq)} y$.
\end{definition}

The refinement condition says the rule cannot reverse comparisons
already settled by $\succeq$. 
It can only make further distinctions
within the indifference classes of $\succeq$.

\begin{definition}[Anonymity]
\label{def:anonymity}
A strict rule $T$ is \emph{anonymous} if for every
$\sigma \in \Sym(\N)$ and every input $(h, \succeq)$ 
it holds that 
$   T(\sigma \cdot h, \sigma \cdot \succeq)
   =
   \sigma \cdot T(h, \succeq).$
\end{definition}

Anonymity is the natural symmetry property. 
Relabeling the input
relabels the output accordingly. 
The rule does not depend on player
identities beyond their roles in $h$ and $\succeq$.

\subsection{Symmetric Inputs}
\label{subsec:symmetric-inputs}

In this subsection we identify the mathematical phenomenon underlying the impossibility. 
The intuitive culprit is clear: when the input cannot
distinguish two tied players, no rule can honestly distinguish them either. 
The technical formulation packages this intuition as the
condition that some intrinsic symmetry of the input exchanges two
players in a common indifference class.

The conceptual content is the following. 
A tied pair of players $x, y$ is potentially honorable to distinguish only if the input
contains some asymmetry between them. 
If, on the other hand, the input is symmetric in $x$ and $y$ (some joint symmetry of $(h,
\succeq)$ exchanges them), there is no information available to
the rule that could justify ranking one above the other. 
We call such an input symmetric. 
Symmetric inputs are the fundamental  obstruction to
strict tie-breaking, and as we shall see (Proposition~\ref{prop:uniformity}
and the discussion that follows), they exist in essentially every
realistic application of the framework.

\begin{definition}[Symmetric Input]
\label{def:symmetric-input}
An input $(h, \succeq)$ is \emph{symmetric} if there exist distinct
$x, y \in \N$ and a permutation $\sigma \in \Stab(h, \succeq)$ with
$\sigma(x) = y$.
\end{definition}

Equivalently, an input is symmetric if and only if the action of
$\Stab(h, \succeq)$ on $\N$ has a non-singleton orbit. 
Note that, by Corollary~\ref{cor:orbits-in-classes}, two such players $x, y$ are
automatically in the same indifference class of $\succeq$.

A symmetric input is one in which two distinct players are intrinsically indistinguishable. 
They are tied in $\succeq$, and a
joint symmetry of $(h, \succeq)$ exchanges them. 
Within $(h, \succeq)$ itself no basis exists for ranking one above the other.

We now identify a clean structural condition on the information
space $\I$ that guarantees the existence of a symmetric input. 
The condition is the existence of a single element of $\I$ that is
invariant under every relabeling of the players. 
We will call such information spaces \emph{uniform}.

\begin{definition}[Uniform Information Space]
\label{def:uniform}
The information space $\I$ is \emph{uniform} if it contains an
element $h_*$ that is fixed by every $\sigma \in \Sym(\N)$, that is,
$\sigma \cdot h_* = h_*$ for all $\sigma \in \Sym(\N)$.
\end{definition}

In group-theoretic language, $\I$ is uniform precisely when the
$\Sym(\N)$-action on $\I$ has a global fixed point.

\begin{proposition}[Uniform Information Spaces Admit Symmetric Inputs]
\label{prop:uniformity}
If $\I$ is uniform and $|\N| \geq 2$, then $\I \times \Wk(\N)$
contains a symmetric input.
\end{proposition}

\begin{proof}
Let $h_* \in \I$ be a global fixed point of the $\Sym(\N)$-action,
so $\sigma \cdot h_* = h_*$ for every $\sigma \in \Sym(\N)$. 
Pick any two distinct $x, y \in \N$ (possible since $|\N| \geq 2$), and
let $\sigma = (x\,y)$ be the transposition swapping $x$ and $y$.

Take $\succeq_\top$ to be the all-tied weak order on $\N$ (one
indifference class). For every $\tau \in \Sym(\N)$ and every
$x', y' \in \N$, we have $x' \,(\tau \cdot \succeq_\top)\, y'$ 
if and only if
$\tau^{-1}(x') \succeq_\top \tau^{-1}(y')$, which always holds since
$\succeq_\top$ is the all-tied weak order. 
So $\tau \cdot \succeq_\top \,=\, \succeq_\top$ for every $\tau$, in
particular for $\sigma$. 
Combined with $\sigma \cdot h_* = h_*$,
this gives $\sigma \in \Stab(h_*, \succeq_\top)$.
Both $x$ and $y$ lie in the unique indifference class of
$\succeq_\top$, namely $\N$ itself. 
They are distinct, and $\sigma(x)= y$ by construction. 
It follows that  $(h_*, \succeq_\top)$ is a symmetric input.
\end{proof}

\begin{remark}[Verification for the Worked Examples]
\label{rem:five-examples-uniform}
Each of the five information spaces of
Subsection~\ref{subsec:examples} is uniform. 
We exhibit a global
fixed point $h_*$ in each.
\begin{itemize}[leftmargin=2em]
\item Example~1 (chess): $h_*$ is the empty tournament record,
$h_*(x, y, t) = (\mathrm{n}, -)$ for all $x, y, t$. 
No games played; relabeling the players does nothing.
\item Example~2 (round-robin): $h_*$ is the all-zeros record,
$h_*(x, y) = 0$ for all distinct $x, y$. 
Every match a draw with zero goal difference; 
every relabeling preserves this.
\item Example~3 (voting): $h_*$ is the zero histogram, $h_*(\succeq')
= 0$ for every $\succeq' \in \Wk(\N)$. 
No votes cast at all;
relabelings act on empty data.
\item Example~4 (cooperative game): $h_*$ is the zero
characteristic function, $v(S) = 0$ for every $S \subseteq \N$. 
Every coalition is worth nothing; symmetric in every sense.
\item Example~5 (network centrality): $h_*$ is the empty graph
$(\N, \emptyset)$, with no edges. 
Relabeling the vertices
preserves the absence of edges.
\end{itemize}
In each case, the fixed point is fixed by every $\sigma$ trivially:
the data records ``nothing'' (or, for Example~2, ``the same nothing
for every pair''), and there is no asymmetric content to disturb.
By Proposition~\ref{prop:uniformity}, the information space of
each example admits a symmetric input.
\end{remark}

\begin{remark}[The Meta-Claim of Uniformity in Realistic Applications]
\label{rem:meta-uniformity}
Proposition~\ref{prop:uniformity} reduces the question of whether
the impossibility theorem (proved in Section~\ref{sec:imp}) applies to a given concrete instance of
the framework to the question of whether the underlying information space is uniform. 
We have just verified that all five worked
examples have uniform information spaces. 
We now make a
\emph{meta-claim}, distinct in character from the formal
mathematical results of this paper.

\smallskip\noindent\textit{Meta-claim.} Essentially every
information space arising in a realistic application of the
framework is uniform.

\smallskip We emphasize that this is a meta-claim and not a
theorem. The notion ``realistic application'' is not a mathematical
predicate, so the meta-claim cannot be formalized, let alone
proved.

Note, however, that the meta-claim is well-supported by structural considerations. 
Indeed, a realistic information space is built from
canonical set-theoretic constructions on the player set $\N$: 
it records functions on $\N$ or on $\N^k$, subsets of $\N$, relations
on $\N$, multisets over data parameterized by $\N$, graphs on $\N$, and so on. 
Every such construction admits a degenerate or
``trivial'' element, namely a constant function, the empty subset,
the empty relation, the zero multiset, the empty graph. 
Such elements, by their very degeneracy, treat all players uniformly,
and are therefore fixed by every $\sigma \in \Sym(\N)$. 
They witness uniformity. 
The five worked examples all fit this pattern, 
and we are not aware of any naturally arising information
space that does not.

It is possible to construct artificial information spaces that
fail to be uniform, for instance by deliberately removing all
$\Sym(\N)$-fixed elements from a naturally constructed space. 
Such constructions are technical curiosities, not models of practical
tie-breaking situations.
In every realistic modeling context we
know of, uniformity is automatic.

We therefore conclude, with the appropriate qualification: the
impossibility theorem proved below applies, in practice, to
essentially every concrete instantiation of the framework. 
The formal scope of the theorem is governed by
Proposition~\ref{prop:uniformity}; the practical scope, by the
meta-claim above.
\end{remark}

\subsection{The Impossibility Theorem}

We are now ready for the first main result. 
The theorem states that the existence of even one symmetric input rules out the existence of
any anonymous strict tie-breaking rule. 
The argument is short and turns on a rather basic group-theoretic fact, 
namely, that a non-trivial permutation
cannot fix any linear order on $\N$. 
The geometric picture is that
asking a strict rule to be invariant under a non-trivial symmetry of
its input is asking it to be a fixed point of an action that has no
fixed points in the codomain.

\begin{theorem}[Impossibility of Strict Anonymous Tie-Breaking]
\label{thm:impossibility}
Let $\I$ be an information space on $\N$, and suppose
$\I \times \Wk(\N)$ contains a symmetric input. 
Then no anonymous
strict tie-breaking rule
$ \I \times \Wk(\N) \to \Lin(\N)$ exists.
\end{theorem}

\begin{proof}
Let $(h_*, \succeq_*)$ be a symmetric input, witnessed by distinct
$x, y \in \N$ and $\sigma \in \Stab(h_*, \succeq_*)$ with
$\sigma(x) = y$. 
By contradiction, suppose that $T$ is an anonymous strict rule. 
Applying anonymity at $\sigma$, 
we obtain
\[
   T(\sigma \cdot h_*, \sigma \cdot \succeq_*)
   =
   \sigma \cdot T(h_*, \succeq_*).
\]
By the choice of $\sigma$ as a member of $\Stab(h_*, \succeq_*)$,
we have $\sigma \cdot h_* = h_*$ and
$\sigma \cdot \succeq_* \,=\, \succeq_*$. 
Substituting,
\[
   T(h_*, \succeq_*) = \sigma \cdot T(h_*, \succeq_*).
\]
Write $L = T(h_*, \succeq_*) \in \Lin(\N)$. 
The above identity
says $L = \sigma \cdot L$, that is, 
$L$ is fixed by $\sigma$.

We now show that no linear order can be fixed by $\sigma$. 
Since $\sigma(x) = y \neq x$, the linear order $L$ contains either
$x >_L y$ or $y >_L x$ (exactly one, by totality and antisymmetry).
Suppose first that $x >_L y$. 
By the definition of the action,
$a \, (\sigma \cdot L) \, b$ holds if and only if
$\sigma^{-1}(a) >_L \sigma^{-1}(b)$. 
Setting $a = \sigma(x)$ and
$b = \sigma(y)$, this gives
$\sigma(x) \,(\sigma \cdot L)\, \sigma(y)$
if and only if $x >_L y$,
which holds by assumption. 
Hence, $\sigma(x) >_{\sigma \cdot L} \sigma(y)$. 
Using $L = \sigma \cdot L$,
we conclude that $\sigma(x) >_L \sigma(y)$, i.e., 
$y >_L \sigma(y)$.
Iterating the same argument starting from $\sigma^k(x) >_L
\sigma^{k+1}(x)$ yields $\sigma^{k+1}(x) >_L \sigma^{k+2}(x)$ for each $k \geq 0$. 
The orbit of $x$ under $\langle \sigma \rangle$ is
finite, say of size $n$, so $\sigma^n(x) = x$. 
Chaining these strict comparisons, we get 
\[
   x \;>_L\; \sigma(x) \;>_L\; \sigma^2(x) \;>_L\; \cdots
   \;>_L\; \sigma^{n-1}(x) \;>_L\; \sigma^n(x) = x.
\]
By transitivity of $L$, we obtain $x >_L x$, contradicting
irreflexivity of the strict part of $L$.

The case $y >_L x$ is handled symmetrically, by starting from
$\sigma(x) >_L x$ and doing a similar iteration. 
\end{proof}

\begin{corollary}[Plain Statement]
\label{cor:plain}
There is no anonymous function $T \colon \I \times \Wk(\N) \to \Lin(\N)$
that takes every input to a refining linear order, whenever the
input space admits a symmetric input.
\end{corollary}

Combining the impossibility theorem with the uniformity
machinery developed in Subsection~\ref{subsec:symmetric-inputs}
gives a clean structural statement.

\begin{corollary}[Uniform Information Spaces Admit No Anonymous Strict Rule]
\label{cor:uniform-impossibility}
If $|\N| \geq 2$ and $\I$ is uniform in the sense of
Definition~\ref{def:uniform}, then no anonymous strict
tie-breaking rule $T \colon \I \times \Wk(\N) \to \Lin(\N)$ exists.
\end{corollary}

\begin{proof}
By Proposition~\ref{prop:uniformity}, $\I \times \Wk(\N)$ contains
a symmetric input. By Theorem~\ref{thm:impossibility}, no anonymous
strict tie-breaking rule on $\I \times \Wk(\N)$ exists.
\end{proof}

In particular, by
Remark~\ref{rem:five-examples-uniform}, all five worked examples
of Subsection~\ref{subsec:examples} have uniform information
spaces, so no anonymous strict tie-breaking rule on any of them exists. 
The meta-claim of Remark~\ref{rem:meta-uniformity} extends this conclusion, 
with appropriate qualification, to essentially every realistic instance
of the framework.

The proof shows that the obstruction is not subtle. 
A single symmetric input, anywhere in $\I \times \Wk(\N)$, suffices to
prevent the existence of any anonymous strict rule.

\begin{remark}[The Strength of the Result]
\label{rem:strength}
Theorem~\ref{thm:impossibility} is sharper than the parametric
ANR-impossibility (anonymity, neutrality, resolvability) results
in the social-choice literature in two respects. 
First, the conclusion (``no anonymous strict rule
exists'') holds unconditionally once a single symmetric input is
present; there is no possibility of escape via favorable arithmetic
on $|\N|$ or $|\I|$. 
Second, the argument is uniform across all
$\Sym(\N)$-sets $\I$, with no special structural assumptions on $\I$. 
This is what allows the same theorem to apply to chess
tournaments, voting, sports leagues, cooperative games, and graphs
all at once.
\end{remark}

The result also admits a useful taxonomy of escape strategies. 
The goal of fair, anonymous, deterministic, strict tie-breaking is
incoherent: the four desiderata cannot be jointly satisfied. 
Any tie-breaking method that produces a strict output, in any of the
realistic application areas, must give up at least one of these.
Either it tolerates ties (giving up strictness), or it depends on
player identity (giving up anonymity), or it uses external
randomization (giving up determinism), or it uses external data (broadening the input). 
Real-world systems make different choices among these escapes. 
Chess Swiss-system cascades give up determinism at the bottom of the cascade, 
drawing lots when the cascade is exhausted. 
Sports leagues give up determinism similarly,
or fall back to fair-play points and other expanded data. 
The Isner--Mahut rule change at Wimbledon gave up the ambition of
deciding every match on the strength of the match record alone,
accepting that the playoff at 12--12 would supply the winner once
the cap is reached. 
Voting tie-breakers vary across all four
escapes depending on the system. 
The framework we develop here
makes this taxonomy of escapes precise and gives a clean
structural account of which information is forced and which is
chosen.

\section{The Partition Framework and the Characterization Theorem}
\label{sec:partition}

The impossibility theorem clearly motivates a reformulation. 
If strict output is unattainable, the natural alternative is to allow the
rule to leave ties unbroken when the input cannot resolve them. 
We make this precise by replacing the codomain $\Lin(\N)$ with
$\PartN$: the rule outputs a partition of $\N$ recording which
players are ``still tied'' and which are ``separated.''

\subsection{Partition-Valued Tie-Breaking Rules}

The conceptual move from a strict to a partition-valued rule is the following. 
Instead of forcing the output to be a complete linear
order, with all ties broken, we allow the rule to declare some pairs
of players ``still tied.'' 
The output is then a partition of $\N$
into blocks of mutually-tied players, refining the indifference partition $\Part(\succeq)$. 
Concretely, each block of the output
records a set of players the rule cannot honestly separate, 
given the information available.

This reformulation is, in fact, not a retreat. 
Rather, it is a recognition that
the natural codomain for a tie-breaking rule, when one wants to record
exactly what the input determines and no more, is the lattice of
partitions of $\N$. 
As we shall see, in this codomain the
impossibility dissolves and is replaced by a clean uniqueness
characterization.

\begin{definition}[Partition-Valued Tie-Breaking Rule]
\label{def:rule-partition}
A \emph{partition-valued tie-breaking rule} is a function
$T \colon \I \times \Wk(\N) \to \PartN$ satisfying the refinement
constraint $T(h, \succeq) \trianglelefteq \Part(\succeq)$ for every
input.
\end{definition}

The refinement constraint says the rule never undoes existing
strict comparisons in $\succeq$. 
Players strictly compared by
$\succeq$ are in different blocks of $T(h, \succeq)$. 
Players tied
by $\succeq$ may end up tied or separated.

\begin{definition}[Anonymity for Partition-Valued Rules]
\label{def:anonymity-partition}
A partition-valued rule
\[
   T \colon \I \times \Wk(\N) \to \PartN
\]
is \emph{anonymous} if for every $\sigma \in \Sym(\N)$ and every input
$(h, \succeq)$ it holds that
\[
   T(\sigma \cdot h, \sigma \cdot \succeq)
   =
   \sigma \cdot T(h, \succeq).
\]
\end{definition}

This is exactly the analogue, for partition outputs, of
Definition~\ref{def:anonymity}. The action on the codomain is
the canonical $\Sym(\N)$-action on $\PartN$ introduced in
Subsection~\ref{subsec:Sym-actions}.

\subsection{The Orbit Partition}

We come now to the central object of the partition framework, namely  a
canonical partition of $\N$ extracted directly from the joint
stabilizer of an input. 
Recall that the joint stabilizer
$\Stab(h, \succeq)$ records the intrinsic symmetries of the input.
Such a group acts on $\N$, and its orbits give a partition of $\N$.
Two players land in the same orbit if and only if some intrinsic
symmetry of $(h, \succeq)$ exchanges them, that is, if and only if
the input contains no asymmetric information distinguishing them.

The orbit partition is, in this precise sense, the canonical record
of which players are intrinsically indistinguishable in the input.
It is the natural partition compatible with the input's symmetries.
Its blocks are exactly the equivalence classes of the relation ``$x$
and $y$ are exchangeable under some symmetry of $(h, \succeq)$,''
and the characterization theorem we prove below shows it is the
unique partition-valued rule satisfying two natural axioms.

\begin{definition}[Orbit Partition]
\label{def:orbit-partition}
For an input $(h, \succeq)$, the \emph{orbit partition}
$\Omega(h, \succeq)$ is the partition of $\N$ whose blocks are the
orbits of $\Stab(h, \succeq)$ acting on $\N$.
\end{definition}

It is a standard fact of group theory that orbits of a group action
on a set form a partition of that set, so $\Omega(h, \succeq)$ is
well-defined as an element of $\PartN$. 
Two players $x, y$ are in
the same block of $\Omega(h, \succeq)$ 
if and only if
 some intrinsic symmetry
of $(h, \succeq)$ exchanges them. 
The orbit partition records
exactly the equivalence relation ``$x$ and $y$ are
indistinguishable in $(h, \succeq)$.''

\begin{lemma}[Refinement of $\Part(\succeq)$]
\label{lem:Omega-refines}
$\Omega(h, \succeq) \trianglelefteq \Part(\succeq)$ for every input
$(h, \succeq)$.
\end{lemma}

\begin{proof}
Let $B$ be a block of $\Omega(h, \succeq)$, i.e., $B$ is a
$\Stab(h, \succeq)$-orbit on $\N$. 
Pick any $x \in B$. 
By
Corollary~\ref{cor:orbits-in-classes}, the orbit of $x$ under
$\Stab(h, \succeq)$ is contained in the indifference class of $x$
under $\succeq$. 
So $B$ is contained in this indifference class,
which is a block of $\Part(\succeq)$. 
Hence, every block of
$\Omega(h, \succeq)$ is contained in a block of $\Part(\succeq)$,
which is the definition of refinement.
\end{proof}

The next lemma records the equivariance of $\Omega$ under the
diagonal $\Sym(\N)$-action on inputs. 
This is the property that, in
the language of the next section, says $\Omega$ is anonymous as a
partition-valued rule, and it is one of the two ingredients used in
the characterization theorem (Theorem~\ref{thm:characterization}) to
identify $\Omega$ uniquely.

\begin{lemma}[Equivariance of $\Omega$]
\label{lem:Omega-equivariant}
$\Omega(\sigma \cdot h, \sigma \cdot \succeq)
 = \sigma \cdot \Omega(h, \succeq)$
for every $\sigma \in \Sym(\N)$ and input $(h, \succeq)$.
\end{lemma}

\begin{proof}
We first show
$\Stab(\sigma \cdot h, \sigma \cdot \succeq)
 = \sigma \, \Stab(h, \succeq) \, \sigma^{-1}$.
For any $\tau \in \Sym(\N)$, it holds that
$\tau \in \Stab(\sigma \cdot h, \sigma \cdot \succeq)$ 
if and only if
$\tau \cdot (\sigma \cdot h) = \sigma \cdot h$ and
$\tau \cdot (\sigma \cdot \succeq) = \sigma \cdot \succeq$. 
By associativity of the action, these conditions are equivalent to
$(\sigma^{-1} \tau \sigma) \cdot h = h$ and
$(\sigma^{-1} \tau \sigma) \cdot \succeq \,=\, \succeq$, i.e.,
$\sigma^{-1} \tau \sigma \in \Stab(h, \succeq)$, equivalently
$\tau \in \sigma \, \Stab(h, \succeq) \, \sigma^{-1}$.

Now consider the orbits of these two stabilizers on $\N$. 
The orbits of $\sigma \, \Stab(h, \succeq) \, \sigma^{-1}$ are exactly
the $\sigma$-images of the orbits of $\Stab(h, \succeq)$, since
\[
   \sigma \, \Stab(h, \succeq) \, \sigma^{-1} \cdot \sigma(x)
   =
   \sigma \cdot \Stab(h, \succeq) \cdot x.
\]
So the partition of $\N$ into orbits of
$\Stab(\sigma \cdot h, \sigma \cdot \succeq)$ is the
$\sigma$-image of the partition into orbits of
$\Stab(h, \succeq)$, i.e.,
$\Omega(\sigma \cdot h, \sigma \cdot \succeq) = \sigma \cdot
\Omega(h, \succeq)$.
\end{proof}

\subsection{The Two Axioms}

We now formulate the axiomatic characterization of the orbit partition rule. 
The intuition is the following. 
A partition-valued rule should
refine $\succeq$ as much as the symmetries of $(h, \succeq)$ permit
and no further.
 ``As much as permit'' is one constraint; ``no further'' is another. 
 Each constraint can be stated independently, and we shall
show that the orbit partition rule is the unique rule satisfying both.

The two axioms have intuitive content beyond their formal statement.
The first, symmetry-saturation, embodies the demand that intrinsically
indistinguishable players be assigned to the same block. 
The rule must be at least as coarse as the symmetries dictate. 
The second, maximal-fineness, embodies the dual demand that the rule make as many
distinctions as possible. 
The rule must be at least as fine as the
input distinguishes. 
Together, the two axioms pin down the partition
exactly: not too coarse, not too fine, but at the unique level where
distinguishability matches the input's information content.

\begin{definition}[Symmetry-Saturation]
\label{ax:sat}
A partition-valued rule $T$ is \emph{symmetry-saturated} if
$\Omega(h, \succeq) \trianglelefteq T(h, \succeq)$ for every input.
\end{definition}

Symmetry-saturation says that if two players are connected by some
intrinsic symmetry of $(h, \succeq)$, they end up in the same
block of $T(h, \succeq)$. 
Equivalently, no symmetry-related pair is
separated by the rule. 
The rule respects the input's symmetries.

\begin{definition}[Maximal-Fineness]
\label{ax:fine}
A partition-valued rule $T$ is said to be 
\emph{maximally fine} if
$T(h, \succeq) \trianglelefteq \Omega(h, \succeq)$ for every input.
\end{definition}

Maximal-fineness says that if two players are not connected by any
intrinsic symmetry, they end up in different blocks. 
The rule does not artificially merge players the input distinguishes.

\begin{remark}[Why Two Axioms]
\label{rem:why-two}
The two axioms pull in opposite directions from each other. 
Symmetry-saturation
imposes a lower bound on the rule (the partition cannot be too fine, because
forcing symmetric pairs into the same block coarsens it).
Maximal-fineness imposes an upper bound (the partition cannot be
too coarse).
They are independent. 
The trivial rule
$T(h, \succeq) := \Part(\succeq)$ satisfies symmetry-saturation but
not maximal-fineness when $\Omega \neq \Part(\succeq)$; the
discrete rule $T(h, \succeq) := \{\{x\} \mid x \in \N\}$ satisfies
maximal-fineness but not symmetry-saturation when $\Omega$ has any
non-singleton block. 
Their conjunction will pin down a unique rule.
\end{remark}

\subsection{The Characterization Theorem}

We can now state and prove the second main theorem. 
The two axioms introduced above are individually weak; 
each is satisfied by many rules. 
Yet together they admit a unique solution: the orbit
partition rule $\Omega$. 
The proof is short and turns on the
antisymmetry of refinement, which is what makes the conjunction of
upper and lower bounds determine a unique partition.

The theorem is positive in spirit, complementing the impossibility
of the previous section. 
Where the strict codomain admits no
anonymous rule, the partition codomain admits a unique honest one.
The orbit partition is, as it were, the canonical answer to the
question ``what does the input determine, given that strict output
is unattainable?''

\begin{theorem}[Characterization of $\Omega$]
\label{thm:characterization}
A partition-valued tie-breaking rule $T$ satisfies both
symmetry-saturation and maximal-fineness if and only if
$T(h, \succeq) = \Omega(h, \succeq)$ for every input. 
The orbit partition rule $\Omega$ moreover satisfies the refinement
constraint and is anonymous.
\end{theorem}

\begin{proof}
We prove each direction.

$(\Leftarrow)$ Suppose $T(h, \succeq) = \Omega(h, \succeq)$ for every input. 
Then it follows that for every input, 
$\Omega(h, \succeq) \trianglelefteq T(h, \succeq)$ (since both
sides are equal, and refinement is reflexive), so $T$ is symmetry-saturated; 
and similarly $T(h, \succeq) \trianglelefteq \Omega(h, \succeq)$, 
so $T$ is maximally fine.

$(\Rightarrow)$ Suppose $T$ satisfies both axioms. 
Then for every input,
$\Omega(h, \succeq) \trianglelefteq T(h, \succeq)$ (by
symmetry-saturation) and
$T(h, \succeq) \trianglelefteq \Omega(h, \succeq)$ (by
maximal-fineness). 
Refinement is a partial order on $\PartN$, in
particular antisymmetric, so if $\mathcal{P} \trianglelefteq
\mathcal{Q}$ and $\mathcal{Q} \trianglelefteq \mathcal{P}$, then
$\mathcal{P} = \mathcal{Q}$. 
Applying this with
$\mathcal{P} = T(h, \succeq)$ and
$\mathcal{Q} = \Omega(h, \succeq)$ gives
$T(h, \succeq) = \Omega(h, \succeq)$ for every input.

Refinement of $\Omega$ as a partition-valued rule is exactly
Lemma~\ref{lem:Omega-refines}. 
Anonymity of $\Omega$ is exactly
Lemma~\ref{lem:Omega-equivariant}.
\end{proof}

\begin{remark}[On the Role of Anonymity]
\label{rem:anonymity-role}
Anonymity does not appear as a hypothesis of
Theorem~\ref{thm:characterization}; the two axioms are sufficient
on their own. 
Anonymity is a consequence (Lemma
\ref{lem:Omega-equivariant} shows $\Omega$ is anonymous, and
$\Omega$ is the unique rule satisfying the two axioms). 
This may be
surprising at first glance, since anonymity is the natural
fairness axiom to impose. 
But anonymity, applied to
$\sigma \in \Stab(h, \succeq)$, gives only that $\sigma$ permutes
the blocks of $T(h, \succeq)$ as a \emph{set}, not that each block
is individually $\sigma$-invariant. 
The latter is what symmetry-saturation explicitly demands, and it is strictly stronger. 
The following example illustrates the gap.
\end{remark}

\begin{example}[Anonymity Does Not Imply Symmetry-Saturation]
\label{ex:anon-not-sat}
Let $\N = \{a, b, c, d\}$ and let $h$ be such that
$\Stab(h)$ is generated by $\sigma = (a\,b)(c\,d)$. Take
${\succeq}$ to be the all-tied weak order. Then $\Stab(h, \succeq)
= \langle \sigma \rangle$ (since $\succeq$ is fixed by every
permutation), and the orbits of $\sigma$ on $\N$ are $\{a, b\}$ and
$\{c, d\}$. So $\Omega(h, \succeq) = \{\{a, b\}, \{c, d\}\}$.

Consider the partition $\mathcal{P} = \{\{a, c\}, \{b, d\}\}$. 
The permutation $\sigma$ swaps the two blocks of $\mathcal{P}$ as
sets: $\sigma \cdot \{a, c\} = \{b, d\}$ and $\sigma \cdot \{b, d\} = \{a, c\}$. 
Hence, $\sigma \cdot \mathcal{P} = \mathcal{P}$. 
We can extend the assignment $T(h, \succeq) :=
\mathcal{P}$ to an anonymous partition-valued rule on all of
$\I \times \Wk(\N)$ as follows. 
On the $\Sym(\N)$-orbit of
$(h, \succeq)$, define $T(\rho \cdot h, \rho \cdot \succeq) :=
\rho \cdot \mathcal{P}$ for each $\rho \in \Sym(\N)$. 
This is well-defined: if $\rho_1 \cdot (h, \succeq) =
\rho_2 \cdot (h, \succeq)$, then $\rho_2^{-1} \rho_1 \in
\Stab(h, \succeq) = \langle \sigma \rangle$, and since
$\sigma \cdot \mathcal{P} = \mathcal{P}$, we have
$(\rho_2^{-1} \rho_1) \cdot \mathcal{P} = \mathcal{P}$, i.e.,
$\rho_1 \cdot \mathcal{P} = \rho_2 \cdot \mathcal{P}$. 
Outside the $\Sym(\N)$-orbit of $(h, \succeq)$, define $T$ to be the orbit
partition rule $\Omega$ (which is anonymous, by
Lemma~\ref{lem:Omega-equivariant}). 
The resulting $T$ is anonymous
on all of $\I \times \Wk(\N)$, refines $\Part(\succeq)$ at every
input, and outputs $\mathcal{P}$ at the input $(h, \succeq)$.

This rule $T$ is anonymous, but it violates symmetry-saturation
and maximal-fineness at the input $(h, \succeq)$. 
Furthermore, 
$T(h, \succeq) = \mathcal{P}$ is neither finer than
$\Omega(h, \succeq)$ (since $\{a, c\}$ is not contained in either
$\{a, b\}$ or $\{c, d\}$) nor coarser than $\Omega(h, \succeq)$
(since $\{a, b\}$ is not contained in either $\{a, c\}$ or $\{b, d\}$). 
The partition $\mathcal{P}$ ``crosses'' the orbit structure. 
It respects the symmetry as a global feature, in the
sense of being $\sigma$-invariant as a set of subsets, without
aligning with the symmetry's orbits.
\end{example}

The characterization theorem gives the partition framework a
philosophical robustness that the strict framework lacks. 
In the strict framework, the impossibility of anonymity forces the
practitioner into a trilemma: tolerate ties, give up anonymity, or
appeal to external data. 
The partition framework offers a clean
fourth option: be honest about the residual ties, leaving them
unresolved in the output and letting the practitioner decide
separately how to break them (or not). 
The output of the canonical partition rule is then 
a complete record of what the input
determines, no more and no less.

This perspective has a methodological consequence. 
A practitioner who wants a strict ranking must, by Theorem~\ref{thm:impossibility},
appeal to data outside $(h, \succeq)$. 
The partition framework
clarifies what is honestly available before this external appeal:
the orbit partition $\Omega$. 
Any strict ranking the practitioner
reports is then an extension of $\Omega$, with the data outside
$(h, \succeq)$ contributing only to the within-block and
between-block ordering. 
The decomposition theorem in the next
section makes this picture exact.

\section{The Decomposition Theorem}
\label{sec:decomp}

The impossibility and characterization theorems address two
different codomains, namely linear orders
 (where no rule exists) and partitions
 (where $\Omega$ is the unique rule). 
The decomposition theorem then unifies them.
Every reasonable strict tie-breaking rule factors uniquely as the
canonical orbit partition followed by an arbitrary completion.

\subsection{Partition-Consistency}

For the decomposition theorem to take a clean form, we restrict
attention to strict rules whose output is structurally compatible
with the canonical orbit partition. 
The relevant compatibility condition, partition-consistency, 
asks that the blocks of the orbit
partition appear as contiguous intervals in the strict rule's output.
This is the natural condition; indeed,  a strict rule that violates it would
split some orbit-block $B$ around the elements of another
orbit-block, placing some members of $B$ above all of the other
block and other members of $B$ below all of it. 
The rule would then fail to treat orbit-blocks as cohesive units, even though the
canonical structure marks the entire block as a single
indistinguishability class. 
Partition-consistency rules out this
anomalous interleaving and captures the rules that genuinely use the
canonical partition as their organizing backbone.

\begin{definition}[Consistency of a Linear Order with a Partition]
\label{def:consistency}
A linear order ${>} \in \Lin(\N)$ is \emph{consistent} with a
partition $\mathcal{P} \in \PartN$ if every block of $\mathcal{P}$
is an interval in $>$, i.e., 
 for every block $B \in \mathcal{P}$ and
every $z \in \N$, if there exist $x, y \in B$ with $x > z > y$,
then $z \in B$.
\end{definition}

A linear order is consistent with a partition if and only if the
blocks appear as contiguous chunks in the linear order. 
The linear order on $\N$ corresponds to a linear order on the blocks of
$\mathcal{P}$ (block $B_1$ above block $B_2$ if every element of
$B_1$ is above every element of $B_2$) together with a linear order
within each block.

\begin{definition}[Partition-Consistent Strict Rule]
\label{def:partition-consistent}
A strict rule $T \colon \I \times \Wk(\N) \to \Lin(\N)$ is
\emph{partition-consistent} if, for every input $(h, \succeq)$, the
linear order $T(h, \succeq)$ is consistent with the orbit
partition $\Omega(h, \succeq)$.
\end{definition}

A strict rule failing partition-consistency would split some
orbit-block $B$ in the output. 
There would exist $x, y \in B$ and a
player $z \notin B$ with $x > z > y$, so that the elements of $B$
are not consecutive in the rule's output. 
This means $B$, an indistinguishability class of the canonical structure, is broken up
across two regions of the output, with members of another block placed in between. 
There is no structural justification for such an interleaving. 
Partition-consistency rules it out.

\begin{remark}[Partition-Consistency and Anonymity]
\label{rem:cons-vs-anon}
Anonymity, in the strict setting, is impossible whenever symmetric
inputs are present (Theorem~\ref{thm:impossibility}).
Partition-consistency, in contrast, is achievable. 
It is the strongest natural compatibility between a strict rule and the
canonical orbit structure that survives the impossibility, and it
is satisfied by every concrete strict tie-breaker that can be
constructed by linearizing the orbit partition.
\end{remark}

\subsection{Completions}

The orbit partition $\Omega$ records what the input determines. 
To turn it into a strict linear order, two further pieces of data are
required: an order within each block, and an order on the blocks
within each indifference class. 
We call this data a completion.

The conceptual role of a completion is to supply exactly the
information that the input itself does not provide. 
Within an orbit block, by construction, no symmetry of $(h, \succeq)$ distinguishes
the members; an order on them is therefore arbitrary in the strongest
sense, requiring data outside $(h, \succeq)$. 
Across blocks within an
indifference class, the situation is similar. 
The blocks have been
isolated as separate orbits, so the input distinguishes them, but
since they are tied at the level of $\succeq$, the strict order
between them must come from somewhere else. 
Both pieces of data are
exactly the irreducibly arbitrary content of strict tie-breaking.

\begin{definition}[Completion]
\label{def:completion}
A \emph{completion} for an input $(h, \succeq)$ is a pair
$(\beta, \gamma)$ consisting of 
\begin{itemize}[leftmargin=2em]
\item a \emph{block-internal order} $\beta$, assigning to each
block $B$ of $\Omega(h, \succeq)$ a linear order on $B$;
\item a \emph{block ordering} $\gamma$, assigning to each
indifference class $E$ of $\succeq$ a linear order on the
set of blocks of $\Omega(h, \succeq)$ contained in $E$.
\end{itemize}
A \emph{global completion} $C$ is a choice of completion
$(\beta_{h, \succeq}, \gamma_{h, \succeq})$ for every input.
\end{definition}

Figure~\ref{fig:two-level} illustrates the two-level partition
structure that a completion is built on top of: 
the weak order $\succeq$ groups $\N$ into indifference classes, and within each
class the orbit partition $\Omega(h, \succeq)$ subdivides further into blocks. 
The data $\gamma$ orders the blocks within each
class, and $\beta$ orders the players within each block.

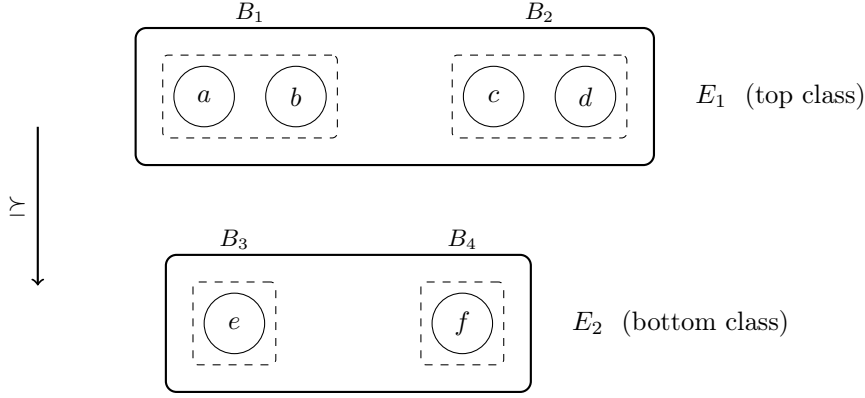
\begin{figure}[h]
\centering
\begin{tikzpicture}[
  player/.style={circle, draw, minimum size=8mm, inner sep=0pt, font=\small},
  block/.style={draw, dashed, rounded corners=2pt, inner sep=4pt, inner ysep=4pt},
  class/.style={draw, thick, rounded corners=4pt, inner sep=10pt, inner ysep=10pt},
]


\node[player] (a) at (0,3) {$a$};
\node[player, right=4mm of a] (b) {$b$};
\node[player, right=18mm of b] (c) {$c$};
\node[player, right=4mm of c] (d) {$d$};

\node[block, fit=(a)(b)] (B1) {};
\node[block, fit=(c)(d)] (B2) {};

\node[font=\footnotesize, above=3mm of B1] {$B_1$};
\node[font=\footnotesize, above=3mm of B2] {$B_2$};

\begin{scope}[on background layer]
  \node[class, fit=(B1)(B2)] (E1) {};
\end{scope}


\node[player] (e) at (0.4, 0) {$e$};
\node[player, right=22mm of e] (f) {$f$};

\node[block, fit=(e)] (B3) {};
\node[block, fit=(f)] (B4) {};

\node[font=\footnotesize, above=3mm of B3] {$B_3$};
\node[font=\footnotesize, above=3mm of B4] {$B_4$};

\begin{scope}[on background layer]
  \node[class, fit=(B3)(B4)] (E2) {};
\end{scope}

\node[font=\small, right=4mm of E1] {$E_1$ \ (top class)};
\node[font=\small, right=4mm of E2] {$E_2$ \ (bottom class)};

\draw[->, thick] (-2.2, 2.6) -- node[left, font=\footnotesize] {$\succeq$} (-2.2, 0.5);

\end{tikzpicture}
\caption{The two-level partition structure on $\N = \{a, b, c, d, e, f\}$. 
The weak order $\succeq$ partitions
$\N$ into indifference classes $E_1$ (top) and $E_2$ (bottom).
Within each class, the orbit partition $\Omega(h, \succeq)$
subdivides further into blocks: $E_1$ contains
$B_1 = \{a, b\}$ and $B_2 = \{c, d\}$; $E_2$ contains
$B_3 = \{e\}$ and $B_4 = \{f\}$. 
The order between classes is
forced by $\succeq$; the order of blocks within each class is the
data $\gamma$; the order of players within each block is the
data $\beta$.}
\label{fig:two-level}
\end{figure}

\begin{definition}[Lift of a Global Completion]
\label{def:lift}
Given a global completion $C = (\beta, \gamma)$, the \emph{lift}
is the rule $T_C \colon \I \times \Wk(\N) \to \Lin(\N)$ defined as follows. 
For input $(h, \succeq)$, the linear order
$T_C(h, \succeq)$ ranks two players $x, y$ as:
\begin{itemize}[leftmargin=2em]
\item if $x$ and $y$ lie in different indifference classes of
$\succeq$, compare them as in $\succeq$;
\item if $x$ and $y$ lie in the same indifference class but in
different blocks of $\Omega(h, \succeq)$, compare them by
$\gamma_{h, \succeq}$ applied to their respective blocks;
\item if $x$ and $y$ lie in the same block $B$, compare them by
$\beta_{h, \succeq}(B)$.
\end{itemize}
\end{definition}

\begin{lemma}[Lift Is Well-Defined and Partition-Consistent]
\label{lem:lift-well-defined}
For every global completion $C$, its lift $T_C$ is a
partition-consistent strict tie-breaking rule.
\end{lemma}

\begin{proof}
We verify the three required properties.

\smallskip\noindent
\emph{$T_C(h, \succeq)$ is a linear order.} 
Reflexivity and totality follow from the construction. 
The three cases of the
definition exhaust all pairs $(x, y) \in \N \times \N$ (as $\Omega(h, \succeq)$ 
refines $\Part(\succeq)$, so the case ``same
indifference class, different blocks'' is well-defined; and the
case ``same block'' is meaningful since blocks are subsets of $\N$). 
For each pair $(x, y)$, exactly one of the three cases
applies, and in each case a strict comparison or equality is
prescribed.

\emph{Antisymmetry.} If $x \neq y$ and the comparison places $x$ above
$y$, then by construction it does not also place $y$ above $x$
(the block-internal $\beta$ is a linear order, the block ordering
$\gamma$ is a linear order, and $\succeq$ is a weak order whose
strict comparisons are antisymmetric).

\emph{Transitivity.} Suppose $x >_T y$ and $y >_T z$. 
We must show $x >_T z$. Let $E_x, E_y, E_z$ denote the indifference classes of
$x, y, z$ in $\succeq$, and $B_x, B_y, B_z$ their orbit blocks.
If $x \succ y$ strictly in $\succeq$ or $y \succ z$ strictly in
$\succeq$, then transitivity of $\succeq$ gives $x \succ z$
strictly, and the first case of the lift gives $x >_T z$. 
We may therefore assume $E_x = E_y = E_z =: E$.

Within $E$, three further subcases:
\begin{itemize}[leftmargin=2em]
\item $B_x = B_y = B_z$: both comparisons $x >_T y$ and $y >_T z$
come from the third case of the lift (same block) and are
statements in the block-internal linear order $\beta(B_x)$.
Transitivity of $\beta(B_x)$ gives $x >_{\beta(B_x)} z$, and
the third case of the lift gives $x >_T z$.
\item $B_x = B_y$, but $B_y \neq B_z$: $x >_T y$ comes from the      
third case (block-internal $\beta$); $y >_T z$ comes from the
second case, giving $B_y >_{\gamma_E} B_z$. 
Since $B_x = B_y$, we have $B_x >_{\gamma_E} B_z$, in particular $B_x \neq B_z$. 
The second case of the lift applied to $(x, z)$ gives $x >_T z$. 
The symmetric subcase ($B_x \neq B_y$, $B_y
= B_z$) is handled identically.
\item $B_x \neq B_y$ and $B_y \neq B_z$: both comparisons come from
the second case, giving $B_x >_{\gamma_E} B_y$ and $B_y >_{\gamma_E} B_z$. 
Transitivity of the strict linear
order $\gamma_E$ yields $B_x >_{\gamma_E} B_z$, in particular $B_x \neq B_z$. 
The second case of the lift gives $x >_T z$.
(In particular, the configuration $B_x = B_z \neq B_y$ is
ruled out: it would force a $\gamma_E$-cycle $B_x >_{\gamma_E}
B_y >_{\gamma_E} B_x$, contradicting irreflexivity of the
strict linear order $\gamma_E$.)
\end{itemize}
In every subcase, $x >_T z$.

\smallskip\noindent
\emph{$T_C(h, \succeq)$ refines $\succeq$.} 
Suppose $x \succ y$ strictly in $\succeq$, i.e., $x$ and $y$ are in different
indifference classes with $x$'s class above $y$'s class. 
Then by the first case of Definition~\ref{def:lift}, $T_C(h, \succeq)$
ranks $x$ above $y$.

\smallskip\noindent
\emph{$T_C(h, \succeq)$ is consistent with $\Omega(h, \succeq)$.}
Each block $B$ of $\Omega(h, \succeq)$ is an interval in
$T_C(h, \succeq)$; by construction, all elements of $B$ are
ranked consecutively within their indifference class, with their
relative order given by $\beta(B)$, and any element outside $B$
is either in a different indifference class (and hence ranked
strictly above or strictly below all of $B$) or in a different
block of the same class (and hence ranked, by $\gamma$, either
strictly above or strictly below all of $B$). 
So no element outside $B$ falls between two elements of $B$, which is the
definition of $B$ being an interval.
\end{proof}

\subsection{The Decomposition Theorem}

We now arrive at the third main result. 
The decomposition theorem says
that every partition-consistent strict tie-breaking rule corresponds
bijectively to a global completion. 
The forced part of any such rule
is the orbit partition $\Omega$, supplied automatically by the framework. 
The free part is the completion, which encodes the rule's
arbitrary choices about how to break the ties that $\Omega$ leaves.
The decomposition theorem makes this informal picture precise.

The conceptual significance of the theorem is the following. 
Tie-breaking systems used in practice (chess cascades, sports lexicographic
rules, voting tie-breakers) are diverse in surface but uniform in structure. 
Each is a specific completion choice on top of a common canonical core. 
The decomposition makes this uniformity explicit and
gives a clean bijection between the universe of partition-consistent
strict rules and the universe of completion data. 
The framework's forced part is shared; 
the diversity lives entirely in the completion.

\begin{theorem}[Decomposition]
\label{thm:decomposition}
The map $C \mapsto T_C$ is a bijection between global completions
and partition-consistent strict tie-breaking rules.
\end{theorem}

\begin{proof}
Lemma~\ref{lem:lift-well-defined} establishes the forward direction:
every global completion lifts to a partition-consistent strict rule. 
We construct the inverse map $T \mapsto C(T)$ and verify
that the two maps are mutual inverses.

\smallskip\noindent
\emph{Construction of the inverse map.} 
Given a partition-consistent strict rule $T$, we define a global
completion $C(T)$ as follows. 
For each input $(h, \succeq)$ and
each block $B$ of $\Omega(h, \succeq)$, set
\[
   \beta^T_{h, \succeq}(B)
   :=
   T(h, \succeq)|_B,
\]
the restriction of the linear order $T(h, \succeq)$ to $B$. 
This is well-defined, since the restriction of a linear order to any subset is a linear order. 
For the block ordering $\gamma^T$, fix an
indifference class $E$ of $\succeq$ and let $B_1, B_2$ be two
blocks of $\Omega(h, \succeq)$ contained in $E$. 
By partition-consistency, $B_1$ and $B_2$ are intervals in $T(h, \succeq)$. 
Pick any $x \in B_1$ and $y \in B_2$, and set
\[
   B_1 >_{\gamma^T_{h, \succeq, E}} B_2
   \quad\iff\quad
   x >_{T(h, \succeq)} y.
\]
This is well-defined (independent of the choice of $x$ and $y$):
if $x' \in B_1$ and $y' \in B_2$, then since $B_1$ and $B_2$ are
intervals not interleaved with each other, either every element
of $B_1$ is above every element of $B_2$ in $T(h, \succeq)$, or
every element of $B_2$ is above every element of $B_1$. 
To see that
$\gamma^T_{h, \succeq, E}$ is a linear order on the blocks of
$\Omega(h, \succeq)$ contained in $E$ 
we argue as follows. 
Irreflexivity is automatic
(distinct blocks have disjoint elements); 
antisymmetry follows since at most one of $x >_T y$, $y >_T x$ can hold for distinct $x, y$; 
transitivity inherits from transitivity of $T(h, \succeq)$; 
and totality follows from totality of $T(h, \succeq)$, since for any two distinct blocks $B_1, B_2$ in
$E$ and representatives $x \in B_1, y \in B_2$, we have $x \neq y$
and either $x >_T y$ or $y >_T x$.

\smallskip\noindent
\emph{$C(T_C) = C$.} Fix a global completion
$C = (\beta, \gamma)$, and consider the lift $T_C$. 
We show that $\beta^{T_C} = \beta$ and $\gamma^{T_C} = \gamma$.

For $\beta^{T_C}$: by definition,
$\beta^{T_C}_{h, \succeq}(B) = T_C(h, \succeq)|_B$. 
By the construction of $T_C$, the restriction of $T_C(h, \succeq)$ to
the block $B$ is exactly $\beta_{h, \succeq}(B)$, so
$\beta^{T_C}_{h, \succeq}(B) = \beta_{h, \succeq}(B)$.

For $\gamma^{T_C}$: by definition, for two blocks $B_1, B_2$ in
the same indifference class $E$,
$B_1 >_{\gamma^{T_C}_{h, \succeq, E}} B_2$ if and only if some
$x \in B_1$ is above some $y \in B_2$ in $T_C(h, \succeq)$. 
By construction of $T_C$, this comparison is precisely
$B_1 >_{\gamma_{h, \succeq, E}} B_2$, so $\gamma^{T_C} = \gamma$.

\smallskip\noindent
\emph{$T_{C(T)} = T$.} 
Fix a partition-consistent strict rule
$T$, set $C = C(T)$, and consider the lift $T_C$. 
We show $T_C(h, \succeq) = T(h, \succeq)$ for every input.
For two players $x, y \in \N$ at input $(h, \succeq)$:
\begin{itemize}[leftmargin=2em]
\item If $x, y$ lie in different indifference classes of
$\succeq$, then $T_C(h, \succeq)$ ranks them as in
$\succeq$, by the first case of the lift. 
 The same is true of $T(h, \succeq)$; indeed,  $T$ refines $\succeq$, so it
ranks $x, y$ as in $\succeq$.
\item If $x, y$ lie in the same indifference class $E$ but in
different blocks $B_x, B_y$ of $\Omega(h, \succeq)$,
then $T_C(h, \succeq)$ ranks them by
$\gamma^T_{h, \succeq, E}$, which by the construction of
$C(T)$ is the comparison induced by $T(h, \succeq)$ on
$B_x$ and $B_y$. 
Since $B_x$ and $B_y$ are intervals in
$T(h, \succeq)$ (by partition-consistency) and not
interleaved, this comparison agrees with the comparison
of $x$ and $y$ directly in $T(h, \succeq)$.
\item If $x, y$ lie in the same block $B$, then
 $T_C(h, \succeq)$ ranks them by $\beta^T_{h, \succeq}(B)
= T(h, \succeq)|_B$, which is precisely the comparison
of $x, y$ in $T(h, \succeq)$.
\end{itemize}
In all three cases, $T_C(h, \succeq)$ and $T(h, \succeq)$ agree
on the comparison of $x$ and $y$. 
Hence, $T_C = T$.
\end{proof}

\begin{remark}[The Unifying Picture]
\label{rem:decomp-unifying}
The three theorems form a coherent picture:
\begin{itemize}[leftmargin=2em]
\item Theorem~\ref{thm:impossibility} identifies the
obstruction to strict tie-breaking.
Anonymity and strict
output are incompatible whenever symmetric inputs exist.
\item Theorem~\ref{thm:characterization} identifies the
canonical positive answer. 
The unique fair partition-valued
rule is the orbit partition.
\item Theorem~\ref{thm:decomposition} ties them together. 
Every partition-consistent strict rule is the orbit partition
followed by a completion.
\end{itemize}
The decomposition theorem makes precise the informal observation
that tie-breaking systems used in practice are honest until forced to be arbitrary. 
The canonical part (the orbit partition) is forced by the framework. 
The completion is the irreducibly arbitrary part,
requiring data not contained in $(h, \succeq)$, for example a coin
flip, a pre-tournament draw position, an alphabetical convention
on players' names, or a fixed seeding. 
Tie-breaking systems used in practice
can now be classified by their completion data: 
a chess cascade is one global completion, a sports lexicographic rule is
another, a pre-tournament-draw lottery is a third. 
Each is a choice of how to fill in the arbitrary part. 
The framework gives us the structural backbone on which all of these sit.
\end{remark}

\subsection{Worked Examples of the Decomposition}

We illustrate the decomposition theorem on two concrete instances.

\begin{example}[Three-Way Symmetric Tie in a Round-Robin]
\label{ex:three-way}
Returning to Example~2, consider a round-robin
tournament with $\N = \{a, b, c, d\}$, in which the final
standings have $a, b, c$ tied at the top with the same points,
and $d$ alone at the bottom. 
The match-result matrix $h$ has:
\begin{itemize}[leftmargin=2em]
\item $a$ beats $b$, $b$ beats $c$, $c$ beats $a$, all by goal
      difference $+1$;
\item each of $a, b, c$ beats $d$ by goal difference $+1$.
\end{itemize}
The cyclic permutation $\rho = (a\,b\,c)$ fixes $h$ (each match
result becomes another match result with the same goal
difference) and fixes $\succeq$ (it permutes the top class
internally and fixes $d$).
 So $\Stab(h, \succeq) \supseteq \langle \rho \rangle$. 
 To compute $\Omega(h, \succeq)$, 
 observe that the orbits of any group between
$\langle \rho \rangle$ and $\Stab(h, \succeq)$ are unions of
$\langle \rho \rangle$-orbits, which are $\{a, b, c\}$ and $\{d\}$.
By Lemma~\ref{lem:stab-preserves-classes}, every
$\sigma \in \Stab(h, \succeq)$ preserves each indifference class of
$\succeq$, hence fixes $\{d\}$ as a set. 
 So $\{d\}$ is a $\Stab(h, \succeq)$-orbit, and $\{a, b, c\}$, already a single
$\langle \rho \rangle$-orbit, is also a single
$\Stab(h, \succeq)$-orbit.
 Therefore, $\Omega(h, \succeq) = \{\{a, b, c\}, \{d\}\}$. 
 Both blocks are intervals in $\succeq$, and in fact $\Omega$ here happens to equal
$\Part(\succeq) = \{\{a, b, c\}, \{d\}\}$.

A partition-consistent strict rule $T$ must produce a linear
order on $\N$ in which $\{a, b, c\}$ is an interval (necessarily
above $\{d\}$, since $T$ refines $\succeq$). 
The completion records:
\begin{itemize}[leftmargin=2em]
\item the block-internal order $\beta(\{a, b, c\})$, which is
      one of the six linear orders on $\{a, b, c\}$;
\item $\beta(\{d\})$ is the trivial order on a singleton;
\item $\gamma$ is trivial, since each indifference class contains
      only one block.
\end{itemize}
So the strict rule's output, modulo the canonical structure, is
encoded entirely in the choice of linear order on $\{a, b, c\}$.
There are six such choices, corresponding to six possible values of
$T(h, \succeq)$ across all partition-consistent strict rules. 
 None of these is more honestly justified than the others.

In real-world football regulations, the choice is typically made
by a lottery, by alphabetical order, or by a pre-tournament draw position. 
The framework's perspective is that each of these is a global completion. 
The lottery completion is randomized; the alphabetical completion uses 
external information (the players' names);
 the draw-position completion uses external information
(the pre-tournament random draw). 
All three are equally arbitrary at this input, and the framework records this
arbitrariness explicitly as the freedom in the completion.
\end{example}

\begin{example}[Two Pairs of Symmetric Voters]
\label{ex:two-pairs}
Returning to Example~3, consider a profile $h$ on
$\N = \{a, b, c, d\}$ consisting of eight ballots: two each of
$a \succ b \succ c \succ d$, $b \succ a \succ c \succ d$,
$a \succ b \succ d \succ c$, and $b \succ a \succ d \succ c$. 
Take $\succeq$ to be the all-tied weak order on $\{a, b, c, d\}$.

The transpositions $(a\,b)$ and $(c\,d)$ both fix $h$ (each
permutes the four ballot types among themselves with multiplicities
preserved), but no permutation exchanging an element of
$\{a, b\}$ with an element of $\{c, d\}$ does, since such a
permutation would move some ballot to a ballot type not present in
the profile. To check this concretely on one such permutation,
consider the transposition $(a\,c)$. It maps the ballot
$a \succ b \succ c \succ d$ to $c \succ b \succ a \succ d$, which
is not among the four ballot types in the profile, so $(a\,c)$
does not fix $h$. The same argument applies, mutatis mutandis, to
any other permutation moving an element of $\{a, b\}$ outside
$\{a, b\}$.
 Therefore, $\Stab(h, \succeq) = \langle (a\,b), (c\,d) \rangle$, with orbits $\{a, b\}$ and
$\{c, d\}$ on $\N$. 
So we have
$ \Omega(h, \succeq) = \{\{a, b\}, \{c, d\}\}, $
a partition with two blocks lying within the single
$\succeq$-class.

A partition-consistent strict rule $T$ must produce a linear order
in which both $\{a, b\}$ and $\{c, d\}$ appear as intervals. 
The completion has both pieces of data non-trivial:
\begin{itemize}[leftmargin=2em]
\item $\gamma$ orders the two blocks within the indifference
      class, placing $\{a, b\}$ above $\{c, d\}$ or vice versa,
      two choices in all;
\item $\beta$ orders within each block, with two choices for
      $\{a, b\}$ and two for $\{c, d\}$, four choices in all.
\end{itemize}
Thus, there are $2 \cdot 4 = 8$ partition-consistent strict rules at
this input, none more honestly justified than the others. 
The example complements the previous one: there, $\gamma$ was trivial
and only $\beta$ played a role; here, both pieces of completion
data contribute.
\end{example}

\section{The Weak-Order Perspective}
\label{sec:weak}

Some readers may prefer to think in terms of weak orders rather than partitions. 
The output of the partition-valued rule $\Omega$ is a
partition of $\N$, and one might wonder whether this partition can
be promoted to a weak order in a natural way, giving a weak-order
analogue of the characterization theorem
(Theorem~\ref{thm:characterization}).

The honest answer is that this promotion is \emph{not} canonical.
A weak order on $\N$ refining $\succeq$ carries, in general, more
information than a partition refining $\Part(\succeq)$. 
A weak order consists of a partition (its indifference partition, which must
refine $\Part(\succeq)$) together with a strict order on the blocks
of this partition within each $\succeq$-class. 
 When the partition is properly finer than $\Part(\succeq)$, this additional ordering data
is non-trivial and not determined by the partition alone. 
The orbit partition $\Omega(h, \succeq)$ provides only the partition data, not the ordering. 
There is, in general, no canonical weak order with
indifference partition $\Omega(h, \succeq)$ refining $\succeq$,
because the choice of how to order $\Omega$-blocks within a single
$\succeq$-class is not determined by $\Omega$ alone.

Nonetheless, the weak-order perspective is informative. 
 The precise relationship between weak orders refining $\succeq$ and partitions
refining $\Part(\succeq)$ is captured by the following lemma, which
makes explicit the additional ordering data needed.

\begin{lemma}[Weak Orders versus Partitions Plus Block Orderings]
\label{lem:weak-bijection}
Fix ${\succeq} \in \Wk(\N)$. There is a bijection between
\begin{enumerate}[leftmargin=2em, label=(\roman*)]
\item the set of weak orders $W \in \Wk(\N)$ that refine $\succeq$,
and
\item the set of pairs $(\mathcal{P}, \delta)$ where $\mathcal{P}$
is a partition of $\N$ refining $\Part(\succeq)$, and $\delta$
assigns to each indifference class $E$ of $\succeq$ a linear order
$\delta_E$ on the set of blocks of $\mathcal{P}$ contained in $E$.
\end{enumerate}
The bijection maps a pair $(\mathcal{P}, \delta)$ to the weak order
$W$ defined by: $x \, W \, y$ if and only if either $x \succ y$
strictly in $\succeq$, or $x \sim_\succeq y$ and the $\mathcal{P}$-block
$B_x$ containing $x$ is above (or equal to) the $\mathcal{P}$-block
$B_y$ containing $y$ in $\delta_E$, where $E$ is the
$\succeq$-class containing both $x$ and $y$. 
The inverse maps $W$ to $(\Part(W), \delta^W)$, 
where $\delta^W_E$ is the order on
$\Part(W)$-blocks within $E$ induced by $W$.
\end{lemma}

\begin{proof}
We verify both directions, 
then check that the two maps are mutual inverses.

\smallskip\noindent
\emph{Forward direction $(\mathcal{P}, \delta) \mapsto W$ is well-defined.} 
Given $(\mathcal{P}, \delta)$, define $W$ as in the statement. 
We verify that $W$ is a weak order, refines $\succeq$,
and has $\Part(W) = \mathcal{P}$.

The relation $W$ is well-defined; indeed,  every pair $x, y \in \N$ falls
into exactly one of three cases (different $\succeq$-classes, same
$\succeq$-class but different $\mathcal{P}$-blocks, or same $\mathcal{P}$-block), 
and in each case the relation is prescribed
unambiguously.

\emph{Reflexivity.} For $x = y$, the same-block case applies
trivially, and $B_x = B_y$ gives $x \, W \, x$.

\emph{Totality.} For $x \neq y$ in different $\succeq$-classes,
totality of $\succeq$ gives $x \succeq y$ or $y \succeq x$, hence
$x \, W \, y$ or $y \, W \, x$. For $x \neq y$ in the same
$\succeq$-class but different $\mathcal{P}$-blocks, totality of
$\delta_E$ gives $B_x \geq_{\delta_E} B_y$ or $B_y \geq_{\delta_E}
B_x$, hence $x \, W \, y$ or $y \, W \, x$. 
For $x, y$ in the same
$\mathcal{P}$-block, both $x \, W \, y$ and $y \, W \, x$.

\emph{Transitivity.} Suppose $x \, W \, y$ and $y \, W \, z$. 
We must show $x \, W \, z$. 
There are nine subcases according to
which case (different classes, same class different blocks, same
block) applies to each of the two given comparisons. 
We treat them together via the following observation. 
The relation $W$, by its definition, is constructed lexicographically:
 first by the strict order on $\succeq$-classes, then within each $\succeq$-class by the
$\delta$-order on $\mathcal{P}$-blocks, then within each
$\mathcal{P}$-block as the trivial all-tied relation. 
Each layer is a transitive (indeed total) relation, and the lexicographic
combination of transitive relations is transitive. 
To verify this claim concretely, 
let $E_x, E_y, E_z$ be the $\succeq$-classes and
$B_x, B_y, B_z$ the $\mathcal{P}$-blocks of $x, y, z$.

\begin{itemize}[leftmargin=2em]
\item If $E_x \succ E_y$ or $E_y \succ E_z$ (strict in $\succeq$),
then $E_x \succ E_z$ strictly. From $x \, W \, y$ we have
$E_x \succeq E_y$, from $y \, W \, z$ we have $E_y \succeq
E_z$, with at least one strict; so $E_x \succ E_z$ strictly,
and the case ``different classes'' applies for $(x, z)$ giving
$x \, W \, z$.
\item Otherwise $E_x = E_y = E_z =: E$. From $x \, W \, y$ we have
$B_x \geq_{\delta_E} B_y$ (with equality if $B_x = B_y$);
from $y \, W \, z$ we have $B_y \geq_{\delta_E} B_z$. By
transitivity of $\delta_E$ it holds that $B_x \geq_{\delta_E} B_z$, so
 $x \, W \, z$.
\end{itemize}

\emph{$W$ refines $\succeq$.} If $x \succ y$ strictly in $\succeq$,
then $E_x \neq E_y$, and the case ``different classes'' applies. 
By construction, $x \, W \, y$. 
We must also check that  $\neg (y \, W \, x)$.
For $y \, W \, x$ to hold, one of the two cases 
of the construction must apply for $(y, x)$: 
either $y \succ x$ strictly in $\succeq$,
which is excluded by $x \succ y$ strictly; 
or $y \sim_\succeq x$,
which is excluded by $E_x \neq E_y$. 
So $\neg (y \, W \, x)$, and $W$
strictly orders $x$ above $y$.

\emph{$\Part(W) = \mathcal{P}$.} For $x, y$ in the same
$\mathcal{P}$-block, both $x \, W \, y$ and $y \, W \, x$ by
construction; so $x \sim_W y$. 
Conversely, if $x, y$ are in
different $\mathcal{P}$-blocks, two subcases:

\begin{itemize}[leftmargin=2em]
\item Different $\succeq$-classes. 
In this subcase, exactly one of
$x \succ y$ or $y \succ x$ holds strictly in $\succeq$ (by
totality and the assumption of different classes). 
By case 1 of the construction, the strict comparison in $\succeq$ is the
one realized by $W$, while the other direction fails (as case
2 of the construction requires same $\succeq$-class). 
Hence, at most one of $x \, W \, y$, $y \, W \, x$ holds, so
 $x \not\sim_W y$.
\item Same $\succeq$-class, different $\mathcal{P}$-blocks. 
By construction, $x \, W \, y$ holds if and only if $B_x \geq_{\delta_E}
 B_y$, and $y \, W \, x$ holds if and only if $B_y \geq_{\delta_E} B_x$.
Both hold if and only if $B_x = B_y$ by antisymmetry of $\delta_E$,
contradicting different blocks. 
So at most one holds, and
$x \not\sim_W y$.
\end{itemize}
Therefore, the indifference classes of $W$ are exactly the
$\mathcal{P}$-blocks, i.e., $\Part(W) = \mathcal{P}$.

\smallskip\noindent
\emph{Inverse direction $W \mapsto (\Part(W), \delta^W)$ is well-defined.} 
Given $W$ refining $\succeq$, define $\delta^W_E$ on
the set of $\Part(W)$-blocks within each $\succeq$-class $E$ by:
$B_1 \geq_{\delta^W_E} B_2$ if and only if some (equivalently,
every) $x \in B_1$ satisfies $x \, W \, y$ for some
(equivalently, every) $y \in B_2$.

\begin{itemize}[leftmargin=2em]
\item $\Part(W)$ refines $\Part(\succeq)$: if $x \sim_W y$, then
neither $x \succ y$ nor $y \succ x$ strictly in $\succeq$
(else $W$ would strictly order them by refinement), so
$x \sim_\succeq y$. 
Hence, the $\Part(W)$-block of $x$ is
contained in the $\Part(\succeq)$-class of $x$, giving refinement.
\item The ``some, equivalently every'' claim. 
Suppose $x \, W \, y$ with $x \in B_1, y \in B_2$, 
where $B_1,B_2$ are blocks of $\Part(W)$, with $B_1 \neq
B_2$. 
Take $x' \in B_1, y' \in B_2$. 
Since $x' \sim_W x$ (same
block) and $y' \sim_W y$ (same block), transitivity of $W$
gives $x' \, W \, x \, W \, y \, W \, y'$, hence $x' \, W \, y'$. 
Furthermore, $y' \, W \, x'$ would by a symmetric
argument give $y \, W \, x$, hence (with $x \, W \, y$)
$x \sim_W y$, contradicting $B_1 \neq B_2$. So $x' \, W \, y'$ strictly.
\item $\delta^W_E$ is a linear order on the set of $\Part(W)$-blocks
 contained in $E$:
\begin{itemize}[leftmargin=2em]
\item \emph{Reflexivity}: for any block $B$, picking $x \in B$
 and using $x \, W \, x$ (reflexivity of $W$) gives
   $B \geq_{\delta^W_E} B$.
   \item \emph{Antisymmetry}: if $B_1 \geq_{\delta^W_E} B_2$ and
   $B_2 \geq_{\delta^W_E} B_1$, then by the
    ``some, equivalently every'' property, picking
    $x \in B_1, y \in B_2$ gives $x \, W \, y$ and
   $y \, W \, x$, hence $x \sim_W y$, so $B_1 = B_2$.
    \item \emph{Transitivity}: if $B_1 \geq_{\delta^W_E} B_2$ and
    $B_2 \geq_{\delta^W_E} B_3$, picking $x \in B_1$, $y \in
   B_2$, $z \in B_3$ gives $x \, W \, y$ and $y \, W \, z$,
    hence $x \, W \, z$ by transitivity of $W$, so
    $B_1 \geq_{\delta^W_E} B_3$.
    \item \emph{Totality}: for any two blocks $B_1, B_2$ within
    $E$, picking $x \in B_1, y \in B_2$ and applying
    totality of $W$ gives $x \, W \, y$ or $y \, W \, x$,
    hence $B_1 \geq_{\delta^W_E} B_2$ or
     $B_2 \geq_{\delta^W_E} B_1$.
     \end{itemize}
    \end{itemize}

\smallskip\noindent
\emph{The two maps are mutual inverses.} Given $(\mathcal{P},
\delta)$, construct $W$, then read off $(\Part(W), \delta^W)$. 
We have $\Part(W) = \mathcal{P}$ by the verification above.
For $\delta^W$, we reason as follows. 
Within each $\succeq$-class $E$, blocks $B_1, B_2$ of
$\mathcal{P} = \Part(W)$ satisfy $B_1 \geq_{\delta^W_E} B_2$ 
if and only if
some $x \in B_1$ has $x \, W \, y$ for some $y \in B_2$, 
if and only if
 $B_1 \geq_{\delta_E} B_2$ by the construction of $W$. 
Hence, $\delta^W = \delta$.

Conversely, given $W$, read off $(\Part(W), \delta^W)$, then construct $W'$. 
By construction, $W'$ ranks pairs in different
$\succeq$-classes by $\succeq$ (matching $W$, which refines
$\succeq$), pairs in the same $\succeq$-class but different
$\Part(W)$-blocks by $\delta^W$ (matching $W$, by the definition
of $\delta^W$), and pairs in the same $\Part(W)$-block as tied
(matching $W$, since $\Part(W)$ is the indifference partition of $W$). 
We conclude that $W' = W$.
\end{proof}

This bijection makes the structural content of weak orders refining $\succeq$ precise. 
Such a weak order consists of a partition (the indifference partition, 
which must refine $\Part(\succeq)$), plus a
linear ordering of the partition-blocks within each $\succeq$-class. 
The orbit partition $\Omega(h, \succeq)$ supplies
the first piece of data canonically; the second piece, the
ordering of $\Omega$-blocks within each $\succeq$-class, is not
determined by $\Omega$ alone and constitutes a genuine choice. 
This is the same choice that appears in the decomposition theorem as
the block-ordering $\gamma$ in a completion.

We may therefore restate the decomposition theorem
(Theorem~\ref{thm:decomposition}) in weak-order language as
follows.

\begin{corollary}[Weak-Order Decomposition]
\label{cor:weak-decomposition}
There is a bijection between
\begin{enumerate}[leftmargin=2em, label=(\roman*)]
\item pairs $(W, \beta)$, where $W$ assigns to each input
      $(h, \succeq)$ a weak order $W(h, \succeq)$ in $\Wk(\N)$
      refining $\succeq$ with indifference partition equal to
      $\Omega(h, \succeq)$, and $\beta$ is a global block-internal
      order in the sense of Definition~\ref{def:completion}, and
\item partition-consistent strict tie-breaking rules.
\end{enumerate}
The forward map sends $(W, \beta)$ to the strict rule $T$ defined
at each input $(h, \succeq)$ as follows. 
For distinct $x, y \in \N$, set $x >_{T(h, \succeq)} y$ if and only if either
\begin{itemize}[leftmargin=2em]
\item $x$ and $y$ lie in different indifference classes of
      $\succeq$ with $x \succ y$ in $\succeq$;
\item $x$ and $y$ lie in the same indifference class of $\succeq$
      but in different $\Omega(h, \succeq)$-blocks, with the block
      of $x$ above the block of $y$ in $W(h, \succeq)$;
\item $x$ and $y$ lie in the same $\Omega(h, \succeq)$-block $B$
      and $x >_{\beta_{h, \succeq}(B)} y$.
\end{itemize}
\end{corollary}

\begin{proof}
By Lemma~\ref{lem:weak-bijection} applied at each input $(h,
\succeq)$ with $\mathcal{P} = \Omega(h, \succeq)$, the data of a
weak order $W(h, \succeq)$ refining $\succeq$ with $\Part(W(h,
\succeq)) = \Omega(h, \succeq)$ is equivalent to the data of a
block-ordering $\gamma_{h, \succeq}$ assigning to each indifference
class $E$ of $\succeq$ a linear order on the $\Omega$-blocks
contained in $E$. 
Hence, pairs $(W, \beta)$ as in (i) are in
bijection with global completions $(\beta, \gamma)$ in the sense of
Definition~\ref{def:completion}. 
By Theorem~\ref{thm:decomposition}, global completions are in
bijection with partition-consistent strict tie-breaking rules.
Composing the two bijections gives the claim.

To verify that the forward map of the corollary agrees with the
forward map of Theorem~\ref{thm:decomposition} (the lift $T_C$), we 
reason as follows. 
Under the bijection of Lemma~\ref{lem:weak-bijection}, the
block-ordering induced by $W(h, \succeq)$ on $\Omega$-blocks within
$E$ is precisely $\gamma_{h, \succeq, E}$. 
The lift's across-block-within-class comparison uses $\gamma$, which by this
correspondence is the across-block comparison induced by $W$. 
The lift's within-block comparison uses $\beta$. 
The lift's across-class comparison uses $\succeq$, 
which is recovered from $W$
since $W$ refines $\succeq$. 
So the lift agrees with the rule
described in the corollary's statement.
\end{proof}

By Lemma~\ref{lem:weak-bijection}, the weak order $W(h, \succeq)$
in part (i) of the corollary is equivalent, at each input, to a
block-ordering $\gamma_{h, \succeq}$ in the sense of
Definition~\ref{def:completion}: $W$ encodes the block-ordering
data $\gamma$ in weak-order language. 
We use whichever notation is
more convenient in the discussion below.

The picture that emerges is a three-level structure:

\begin{itemize}[leftmargin=2em]
\item The orbit partition $\Omega(h, \succeq)$ is canonical, 
forced by the framework, with no choice involved.
\item A weak order $W(h, \succeq)$ refining $\succeq$ with
$\Part(W) = \Omega$ is partially canonical. 
Its indifference structure is forced (it must equal $\Omega$), 
but the strict order on $\Omega$-blocks within each $\succeq$-class is a choice.
\item A linear order $T(h, \succeq)$ refining $W$ resolves all
remaining ties. 
Beyond the structure of $W$, it specifies a
linearization within each $\Omega$-block.
\end{itemize}

The free choices accumulate as we move from one level to the next.
The orbit partition is the canonical core; the weak order adds
block-ordering data $\gamma$; the linear order adds within-block
linearization data $\beta$. 
 Each layer of completion data corresponds
to a layer of arbitrariness in the resulting tie-breaking output.

\begin{remark}[Absence of Anonymity in
Corollary~\ref{cor:weak-decomposition}]
\label{rem:weak-no-anon}
The bijection of Corollary~\ref{cor:weak-decomposition} treats $W$
and $\beta$ as arbitrary functions of the input, with no equivariance condition imposed. 
Correspondingly, the strict rules on
the right-hand side of the bijection need not be anonymous. 
This is the appropriate setting for the corollary, because by
Theorem~\ref{thm:impossibility}, anonymous strict tie-breaking rules
do not exist whenever symmetric inputs are present (and they are
present in essentially every realistic application). 
Imposing anonymity on $W$ and $\beta$ would empty out both sides of the
bijection in any nontrivial setting.

The orbit partition $\Omega$, by contrast, is anonymous in the sense
of Lemma~\ref{lem:Omega-equivariant}; this is exactly why
$\Omega$ alone, viewed as a partition-valued rule, can be
characterized axiomatically (Theorem~\ref{thm:characterization}).
The weak-order and linear-order layers cannot be characterized this
way because they involve the free choices $\gamma$ and $\beta$,
which are precisely the ingredients ruled out by anonymity.
\end{remark}

\section{Related Work}
\label{sec:related}

Our framework was developed independently of the social-choice
literature on what is known as the ANR-impossibility
 (anonymity, neutrality, resolvability), 
  which we describe here for context. 
This voting-specific impossibility, in its various refinements,
is well-established, and several techniques in that literature
overlap with ours.

\subsection{The ANR-Impossibility Tradition}

The ANR tradition addresses the incompatibility of anonymity,
neutrality, and resolvability in voting. The framework is fixed
(voters with linear-order preferences, voting on alternatives), and
the question is for which parameter values $(n, m)$ a rule
satisfying all three properties exists. The answer is given by
arithmetic conditions on the parameters.

\paragraph{Moulin (1983).} Moulin~\cite{moulin1983} proved that, for
voting with $n$ voters and $m$ alternatives where preferences are
linear orders and the decision is a single alternative, an
anonymous, neutral, resolute rule exists if and only if
$\gcd(n, m!) = 1$. 
This is a parametric impossibility: the existence of fair rules
depends on arithmetic conditions on $(n, m)$.

\paragraph{Bubboloni and Gori.} In a sequence of papers from
2014 onward, Bubboloni and Gori extended Moulin's result
substantially: to rankings as decisions~\cite{bubboloni2014}, to
social choice correspondences with refinements~\cite{bubboloni2016},
to multi-winner committees, and to weakened versions of anonymity
and neutrality determined by subgroups~\cite{bubboloni2021}. 
Their framework is voting-specific, with voter-profile spaces and group
actions of $\Sym_h$ on voters and $\Sym_n$ on alternatives, and
their analyses give detailed arithmetic conditions involving subgroups
$V \leq \Sym_h$, $W \leq \Sym_n$.

\paragraph{Xia (2023, 2024).} Xia~\cite{xia2023, xia2024} introduces
the notion of a ``fixed-point decision'' (a decision $d$ such that
$\sigma(d) = d$ for every $\sigma$ stabilizing the histogram of
the profile), shows that ``most equitable refinements'' (rules
satisfying ANR at every non-problematic profile) always exist,
and gives polynomial-time algorithms for computing them. 
The structural insight that a non-trivial permutation cannot
fix a strict ranking is shared with our argument, though our
formulation runs in a setting-independent framework.

\subsection{Other Adjacent Threads}

Beyond the ANR-impossibility tradition, several other strands of
research touch on themes related to ours. 
None of them addresses the abstract tie-breaking problem we study, 
but each shares some mathematical or conceptual machinery, 
and we describe the connections briefly.

\paragraph{Tournament ranking.} 
A separate body of work, much of
it by Csat\'o~\cite{csato2017}, studies axiomatic properties of
specific scoring methods (such as row-sum, generalized row-sum,
least-squares) for tournaments. 
These works are concrete (real-valued scores) and method-specific, where ours is abstract.

\paragraph{Symmetry-based fairness.}
Bartholdi, Hann-Caruthers, Josyula, Tamuz, and
Yariv~\cite{bartholdi2020} study ``equitable'' rules, defined as
those having a transitive automorphism group, with the focus on the
size of winning coalitions in two-alternative voting. 
This is a different question, but the group-theoretic toolkit is similar.

\paragraph{Computational social choice.} 
Many of the ideas
discussed above sit within the broader field of computational
social choice, which studies aggregation procedures from a
combination of axiomatic, algorithmic, and complexity-theoretic perspectives. 
The standard reference for the field is the
\emph{Handbook of Computational Social Choice}~\cite{handbook2016},
edited by Brandt, Conitzer, Endriss, Lang, and Procaccia, whose
chapters on social choice axiomatics, tournament solutions, and
voting tie-breaking touch directly on the questions our framework addresses. 
Our framework is axiomatic in the classical sense rather
than computational, but the questions it raises (which features of a
tie-breaking rule are forced by symmetry, which are free choices)
are of a kind familiar to readers of that volume.

\paragraph{Saari's program.} 
Saari~\cite{saari1995} pioneers the
use of group theory to understand voting paradoxes. 
The focus there is on positional scoring rules and the structure of voting
paradoxes (rather than tie-breaking specifically), but the
spirit (using symmetry to explain aggregation phenomena) is shared.

\paragraph{Connection to Arrow.} 
As noted in the introduction,
the Condorcet cycle in Arrow's framework~\cite{arrow1963} is
structurally a symmetric input in our framework. 
Indeed,  the cyclic permutation maps the cyclic profile to itself, 
forcing anonymity to identify the candidates. 
Arrow's impossibility theorem can
be viewed, in this light, as a voting-specific instance of the
same phenomenon our impossibility theorem captures abstractly.
We do not develop this connection formally here; 
it is one of several directions for further work.

\subsection{Contributions of the Present Paper}

The following ingredients of the present paper are new.
\begin{itemize}[leftmargin=2em]
\item \emph{The framework-free setup.} 
Our framework abstracts to any
$\Sym(\N)$-set $\I$ and applies equally to chess
tournaments, sports leagues, cooperative games, network
centrality, and other settings without voters or alternatives.
\item \emph{The unconditional formulation of the impossibility.}
The earlier results are parametric, with arithmetic conditions
on $(m, n)$ or on subgroup sizes. Our impossibility theorem is
unconditional: a single symmetric input anywhere in
$\I \times \Wk(\N)$ kills strict anonymous tie-breaking,
and symmetric inputs exist in essentially every realistic application.
\item \emph{The partition-valued framework.} 
The earlier tradition outputs linear orders, weak orders, or sets
of decisions. The choice to output partitions of players appears
not to have been studied in the social-choice literature, and it
gives a substantively cleaner characterization theorem.
\item \emph{The characterization of $\Omega$ through
symmetry-saturation and maximal-fineness.} 
The orbit-sta\-bi\-liz\-er machinery is well known and used (notably,
extensively by Xia), but the axiomatic
characterization in terms of these two specific axioms is novel.
\item \emph{The decomposition theorem.} 
The bijection $T \leftrightarrow (\Omega, \text{completion})$ has
not been stated in the existing literature.
\end{itemize}

\section{Conclusion and Further Directions}
\label{sec:conclusion}

We have given an abstract axiomatic theory of tie-breaking. 
The framework rests on a single piece of structure, namely the action of the
symmetric group on an abstract information space, and produces three theorems. 
The impossibility theorem rules out anonymous strict
tie-breaking in any reasonable instance of the framework, 
in a form sharper than the parametric ANR-impossibility results in the social-choice literature. 
A single symmetric input anywhere in the
input space forces the conclusion, with no dependence on parameters
such as the number of players or the size of the input space. 
The characterization theorem identifies the unique fair partition-valued rule as the orbit
partition $\Omega$, the partition into orbits of the joint stabilizer. 
The decomposition theorem expresses every reasonable
strict rule as the canonical partition $\Omega$ plus an arbitrary
completion, in a bijection between completion data and
partition-consistent strict rules.

The conceptual picture that emerges is the following. 
Tie-breaking, considered abstractly, has a forced part and a free part. 
The forced part is the orbit partition: the canonical record of what the input
determines, computed directly from the symmetries of the input. 
The free part is the completion: the arbitrary data needed to linearize
the canonical partition into a strict order. 
Every concrete tie-breaking practice we are aware of (chess Swiss-system cascades,
sports lexicographic regulations, voting tie-breakers, cooperative
game power indices, network centrality measures) sits within this common structure. 
Each such practice can be classified by its completion data, 
and the framework gives a uniform language for comparing and analyzing them. 
Practices that look very different on the surface 
 (a chess Buchholz cascade and a Borda-count tie-breaker)
turn out to share the same canonical core, 
with the differences living entirely in the completion. 
The framework makes this unification precise.

The theory also gives a clean answer to a question that practitioners
of tie-breaking have wrestled with implicitly for as long as
tie-breaking has existed: how much of a tie-breaking method is
forced by the data, and how much is the practitioner's free choice?
Theorem~\ref{thm:decomposition} answers this exactly. 
The forced part is $\Omega$; everything else is free. 
Tie-breaking systems used in practice, in this precise sense, are honest until forced to be arbitrary. 
In fact, there is no honest strict tie-breaking method that escapes this division. 
The practitioner's task, in any concrete setting, is to be honest about
which of their choices belong to the forced canonical core and which
belong to the arbitrary completion.

Several directions for further work suggest themselves.

\paragraph{Probabilistic completions.}
Theorem~\ref{thm:decomposition} characterizes strict completions as
deterministic linearizing data. 
A natural extension considers probability distributions over completions, 
i.e., randomized tie-breaking rules. 
The canonical randomized tie-breaker chooses
uniformly among all completions, 
with each block-ordering $\gamma$ and
each block-internal linearization $\beta$ selected uniformly.
This is the analogue, in our setting, of the random dictatorship
rule in social choice and Gibbard's randomized
mechanisms~\cite{gibbard1977}. 
By construction, this rule recovers anonymity in the marginal sense, 
since under it, the probability that
player $x$ is ranked above player $y$ depends only on the structural
position of $x$ and $y$ in the input (specifically, on whether they
lie in the same $\Omega$-block, the same $\succeq$-class but
different $\Omega$-blocks, or different $\succeq$-classes), not on their identities. 
We expect the analogue of Theorem~\ref{thm:decomposition} to hold for probabilistic rules,
with the canonical part being $\Omega$ and the free part being a
distribution over completions.

\paragraph{Connection to Arrow.} 
Our framework structurally subsumes Arrow's setting: 
the Condorcet cycle, which is the canonical input forcing Arrow's theorem, 
appears as a symmetric input in our
framework, and our impossibility theorem
(Theorem~\ref{thm:impossibility}) reduces to a voting-specific
Arrow-like statement under the appropriate instantiation. 
The informal connection sketched in Section~\ref{sec:related} 
between our tie-breaking impossibility theorem and
Arrow's famous impossibility theorem deserves to be made formal. 
The natural instantiation takes $\I$ to be the space of voter profiles
 (anonymized, as a $\Sym(\N)$-set under candidate relabeling), 
 with $\succeq$ the output of some primary aggregation. 
 Arrow's setting, in which the input is the profile and the output is a social ranking,
 can be recovered by taking $\succeq$ to be the all-tied weak order and
asking the rule to break all ties. 
The Arrow-like conclusion, that no anonymous, neutral, transitive social welfare function exists,
should follow from Theorem~\ref{thm:impossibility} once we identify
the appropriate symmetric inputs. 
The Condorcet cycle is one such input, but not the only one. 
We expect that the Independence of Irrelevant Alternatives
axiom corresponds to a structural compatibility condition between
the rule and certain restrictions of the action, but the exact
formulation requires care. 
Working this connection out explicitly would give a uniform group-theoretic account of Arrow-style
impossibilities, of which our Theorem~\ref{thm:impossibility} would
be the abstract template and Arrow's the voting-specific instance.

\paragraph{Concrete instantiations.} The framework should be tested
on additional concrete settings, beyond the five worked examples
in Section~\ref{subsec:examples}. 
Of particular interest is the lexicographic structure of cascading tie-breakers. 
Real-world practices typically cascade, such as in chess: 
 first compute Buchholz, breaking ties by it;
  if ties remain, compute Sonneborn--Berger and break with
that; if ties still remain, fall through to direct encounter, then to lots. 
Each level of the cascade is a tie-breaking rule in our
sense, and the cascade as a whole has a natural nested-completion structure. 
The orbit partition at level $k$ refines the orbit
partition at level $k-1$, since later levels have access to more auxiliary data. 
A separate axiomatic treatment of cascading
tie-breakers, building on Theorem~\ref{thm:decomposition}, 
would clarify the structure of practical tie-breaking systems and the
sense in which adding a level of the cascade is a refinement operation. 
We expect such a theory to identify which features of a
cascade are intrinsic (forced by the canonical structure at each
level) and which are arbitrary (completion choices). 
Other concrete instantiations of interest include weighted voting (where players
have different voting weights and the symmetric group must be
restricted to weight-preserving permutations), preference
aggregation with abstentions, and team-based tournaments with
internal team structure.

\paragraph{Equivariant completions.}
Theorem~\ref{thm:decomposition} treats completions as fully
arbitrary, with no structural relationship between the completion data at different inputs. 
In practice, tie-breaking methods used in deployed systems
are computed by uniform algorithms whose outputs depend on the
input in a structured way, where changing the input slightly changes the completion slightly. 
A natural restriction on completions is to ask
that they be equivariant under some sub-group of $\Sym(\N)$. 
Full $\Sym(\N)$-equivariance is impossible by
Theorem~\ref{thm:impossibility}, but partial equivariance, for
example under cyclic groups generated by named permutations, or up
to lexicographic adjustments, may be tractable. 
We conjecture that imposing equivariance under various sub-groups gives a stratified
hierarchy of admissible completions, with stronger equivariance
giving stronger structural constraints. 
The interaction between equivariance and computability of the completion is also of
interest, since  a completion that is equivariant under a large sub-group
is, in some sense, more uniform and thus simpler to specify.

\paragraph{Categorical reformulation.} 
The framework has a natural categorical flavor. 
The information space $\I$ is an object in
the category of $\Sym(\N)$-sets, with $\Sym(\N)$-equivariant functions as morphisms. 
The orbit partition $\Omega$ is functorial
in a precise sense: a $\Sym(\N)$-equivariant map
$\phi \colon \I \to \I'$ satisfies
$\Omega(h, \succeq) \trianglelefteq \Omega(\phi(h), \succeq)$ for
every input. That is, applying $\phi$ can only coarsen the orbit
partition, never refine it, since equivariant maps preserve all
symmetries of the source.
The decomposition theorem, properly interpreted, has a natural-transformation structure: 
the assignment ``rule $\mapsto$ completion'' is functorial in a suitable sense.
Whether the framework can be lifted to a categorical theory of
tie-breaking, and whether such a theory connects to other
equivariant frameworks (equivariant homotopy theory, equivariant
$K$-theory, the representation theory of the symmetric group), is
a substantial question that we have not seriously attempted to answer as yet. 
A promising starting point is the observation that the
orbit-partition functor is essentially the construction that sends
a $\Sym(\N)$-set to its quotient by the natural relation, viewed
through the joint stabilizer. 
This is reminiscent of constructions
in (equivariant) algebraic topology where one passes from a $G$-space
to its space of orbits, and the analogy may be more than formal.

\paragraph{Beyond the symmetric group.} 
The framework as stated takes the relevant symmetry group 
to be the full symmetric group $\Sym(\N)$. 
There is no reason in principle that this should be the only choice.
 In settings with structural constraints (teams in
team competitions, parties in proportional elections, fixed
seedings in tournament brackets), the natural symmetry group is a
proper subgroup of $\Sym(\N)$. 
The same theorems as proved here
should go through, mutatis mutandis, for any symmetry group $G
\leq \Sym(\N)$, with $G$-anonymity replacing anonymity, $G$-orbits
replacing $\Sym(\N)$-orbits, 
and $G$-symmetric inputs replacing symmetric inputs. 
The resulting theory would be parametric in $G$
and would interpolate between the fully symmetric setting (when
$G = \Sym(\N)$) and the trivially asymmetric setting (when $G$ is trivial). 
Working out this generalization is straightforward in
principle and may yield additional structural insights for
applications with constrained symmetry.


\end{document}